 \documentclass[preprint,11pt]{elsarticle}

\usepackage{epsfig,multicol,bbm,amsmath,amssymb,amscd,mathrsfs,fancybox,amsthm,framed,enumerate}
\usepackage{txfonts}

\newtheorem{Thm}{Theorem}[section]

\newtheorem{Prop}{Proposition}[section]

\newcommand{\e}{\mathrm{e}}
\newcommand{\supp}{\mathrm{supp\,}}

\newcommand{\Expect}[1]{\left\langle{#1}\right\rangle}
\newcommand{\bvec}[1]{\boldsymbol{#1}}
\newcommand{\no}{\nonumber}

\newcommand{\Natural}{\mathbb{N}}
\newcommand{\Integer}{\mathbb{Z}}
\newcommand{\Real}{\mathbb{R}}

\newcommand{\Complex}{\mathbb{C}}

\newcommand{\bra}[1]{\langle #1 |}
\newcommand{\ket}[1]{| #1 \rangle}

\begin{document}

\begin{frontmatter}



\title{A Note on Reflection Positivity and
the Umezawa-Kamefuchi-K\"all\'en-Lehmann Representation of Two Point Correlation Functions}


\author{Kouta Usui}

\address{Department of physics, the University of Tokyo, 113-0033, Japan\\
Institute for the Physics and Mathematics of the Universe (IPMU),   the University of Tokyo, Chiba 277-8568, Japan\\
Email : kouta@hep-th.phys.s.u-tokyo.ac.jp}

\begin{abstract}
It will be proved that a model of lattice field theories which satisfies \textbf{(A1)} Hermiticity,
\textbf{(A2)} translational invariance, \textbf{(A3)} reflection positivity,
and \textbf{(A4)} polynomial boundedness of correlations, 
permits the Umezawa-Kamefuchi-K\"all\'en-Lehmann representation of two point correlation functions
with a positive spectral density function.
Then, we will also argue that positivity of spectral density functions is necessary for a lattice theory 
to satisfy conditions \textbf{(A1)} -\textbf{(A4)}. As an example,
a lattice overlap scalar boson model will be discussed.
We will find that the overlap scalar boson violates the reflection positivity.
\end{abstract}


\end{frontmatter}


\section{Introduction}
Lattice regularization \cite{Wilson:1974sk} of quantum field theories gives us influential tools 
to analyze the theory on a non perturbative level. It provides 
a mathematically rigorous scheme as well as enables us
to perform numerical simulations, which is really  
powerful when one tackles problems concerning Quantum Chromodynamics (QCD),
the most promising physical model which is expected to describe the strong interaction of elementary 
particles.

However, it is far from trivial whether a lattice scheme indeed defines
a physically acceptable quantum mechanical model. 
M. L\"uscher constructed, starting from Wilson's lattice QCD, a Hilbert space of 
state vectors, and a self-adjoint Hamiltonian operator \cite{Luscher:1976ms}.
This construction is physically natural and concrete, but seems to 
crucially rely on the nearest-neighbor property of the lattice action.
K. Osterwalder and E. Seiler proved \cite{Osterwalder:1977pc} that Wilson's lattice QCD model
fulfills Osterwalder-Schrader's reflection positivity condition \cite{Osterwalder:1973dx,
Osterwalder:1974tc}, and also that form this condition
a Hibert space of quantum mechanical state vectors,
and positive Hamiltonian operator can be reconstructed. 
Their construction is more abstract than Lushcer's
but seems to be applicable to larger class of lattice models,
which may contain non-nearest-neighbor interactions.

Thus, when one considers a lattice model containing
non-nearest-neighbor interactions, one should rely on 
 Osterwalder-Seiler's reconstruction procedure.
In this case, it is an important issue to prove the reflection positivity condition
in order for the lattice model to be ensured to define a quantum mechanical system.
 But, to prove that a concrete lattice model indeed satisfies the reflection 
positivity is not a very trivial problem especially when the lattice action 
contains infinite-range interaction (cf. Ref. \cite{Kikukawa:2010gq}). 
And, such action is really needed when one wants to have the exact chiral symmetry
on the lattice. In fact, the exact chiral symmetry is realized by adopting 
Neuberger's overlap Dirac 
operator \cite{Neuberger:1997fp,Neuberger:1998wv,Luscher:1998pqa}, 
a gauge covariant solution of the Ginsparg-Wilson relation \cite{Ginsparg:1981bj}, and this Dirac 
operator is not finite-range.

Recently, Y. Kikukawa and the author discussed \cite{Kikukawa:2010zy} the $\mathcal{N}=1$ lattice Wess-Zumino model \cite{Wess:1973kz}
formulated through the overlap Dirac operator \cite{Fujikawa:2001ka,Fujikawa:2001ns,Bonini:2004pm,Kikukawa:2004dd,
Bonini:2005qx,Chen:2010uca,Aoyama:1998in}
and pointed out that this model violates the reflection positivity condition.
The reflection positivity condition of this model is violated by the bosonic part, the 
lattice overlap boson. It was shown there that, for the overlap boson, the way to prove the reflection positivity, which
they adopted to prove the reflection positivity of the overlap fermion, 
does not work and that the spectral density function in the Euclidean Umezawa-Kamefuchi-K\"all\'en-Lehmann representation 
\cite{Umezawa:1951rp,Kallen:1952zz,Lehmann:1954xx,Luscher:2000hn}
of the two point correlation function is not
positive. Thus, they concluded that this model does not satisfy the reflection positivity condition.

However, a rigorous proof of the statement that a lattice model with a spectral density
function which is not positive violates the reflection positivity was not given there.
It was assumed there that a lattice model fulfilling the reflection positivity condition permits the
following formal computations :
\begin{align}
\Expect{\phi(x)^*\phi(0)}&=\bra\Omega \hat\phi(x)^\dagger \hat\phi(0) \ket\Omega\no \\
	&=\bra\Omega \{\e^{-Ht-i\bvec P\cdot \bvec x}\hat\phi(0) \e^{Ht+i\bvec P\cdot \bvec x} \}^\dagger\hat\phi(0) \ket\Omega\no \\
	&=\bra\Omega \e^{Ht-i\bvec P\cdot \bvec x}\hat\phi(0)^\dagger \e^{-Ht-i\bvec P\cdot \bvec x} \hat\phi(0) \ket\Omega \no\\
	&=\bra\Omega \hat\phi(0)^\dagger \e^{-Ht-i\bvec P\cdot \bvec x} \hat\phi(0) \ket\Omega \no\\
	&=\int \frac{d\lambda}{\pi} \frac{d\bvec p}{(2\pi)^{d-1}}\,\e^{-\lambda t-i\bvec p\cdot \bvec x}
	\bra\Omega \hat\phi(0)^\dagger \ket{\lambda, \bvec p} \bra{\lambda,\bvec p}\hat\phi(0) \ket\Omega \no\\
	&=\int \frac{d\lambda}{\pi} \frac{d\bvec p}{(2\pi)^{d-1}}\,\e^{-\lambda t-i\bvec p\cdot \bvec x}
	| \bra{\lambda,\bvec p}\hat\phi(0) \ket\Omega |^2\no\\
	&=\int \frac{d\lambda}{\pi} \frac{d\bvec p}{(2\pi)^{d-1}}\,\e^{-\lambda t-i\bvec p\cdot \bvec x}\rho(\lambda,\bvec p).
\end{align} 
to conclude that
reflection positive
models must have the positive spectral density function $\rho(\lambda,\bvec p)=
| \bra{\lambda,\bvec p}\hat\phi(0) \ket\Omega |^2$. 

In this paper, we will give a mathematically rigorous proof of this statement in a self-contained manner.
Assuming that a lattice model of complex scalar field satisfies \textbf{(A1)} Hermiticity, \textbf{(A2)}
translational invariance, \textbf{(A3)} reflection positivity, and \textbf{(A4)} polynomial
boundedness of correlations, we prove that such a lattice model permits the
Euclidean version of Umezawa-Kamefuchi-K\"all\'en-Lehmann representation with a positive spectral
density function. Furthermore, we will point out that positivity of the spectral
density function is necessary for a lattice model to satisfy 
the assumptions \textbf{(A1)}-\textbf{(A4)}. As an application, we will 
prove that the lattice overlap scalar boson violates the reflection positivity
condition by showing that the spectral density function is not positive. 
Therefore, it is somewhat doubtful whether this model really defines
a quantum mechanical model, at least if the lattice spacing is kept non-zero. 

This paper is organized as follows. In section \ref{sec: General setup},
we will introduce several definitions and assumptions \textbf{(A1)}-\textbf{(A4)}
of a generic lattice scalar field model which will be discussed in
the following sections. It will be stated as Theorem \ref{main thm} that a two point
correlation function of a lattice field theory on $\Integer^d$ satisfying 
these assumptions is a Fourier 
(or Laplace, in the time direction) transformation of some positive measure
supported on $[0,\infty)\times[-\pi,\pi]^{d-1}$. This expression is regarded 
as a Euclidean version of the Umezawa-Kamefuchi-K\"all\'en-Lehmann representation, and
from the theorem, the spectral density function is proved to be positive.

In section \ref{sec: lattice QM}, we will review in detail how to reconstruct 
a quantum mechanical system from a lattice field theory satisfying 
assumptions \textbf{(A1)}-\textbf{(A4)}.
These quantum mechanical ingredients will play an essential role in the proof of Theorem \ref{main thm}. 

In section \ref{sec: proof}, Theorem \ref{main thm}
will be proved as a simple corollary of this reconstruction
procedure, and Euclidean version of Umezawa-Kamefuchi-K\"all\'en-Lehmann representation will be discussed.
It will be pointed out there, a lattice model with a non-positive spectral density 
violates at least one of our assumptions \textbf{(A1)}-\textbf{(A4)}.

In section \ref{sec: overlap boson}, the above general discussion will be applied
to the overlap boson. The lattice overlap boson field in the infinite volume lattice will be defined as a
Gaussian random process characterized by the Klein-Gordon type
bosonic overlap operator. After proving that the overlap boson fulfills
all of our assumptions except the reflection positivity condition,
an explicit formula of the two point function
will be presented, which shows that the spectral density is \textit{not}
non-negative.
Finally, this proves that the lattice overlap boson system never satisfies the reflection
positivity condition \textbf{(A3)}. 

\section{General setup}\label{sec: General setup}
We give a basic setup of a lattice field theory considered in the
following. 
Here, we deal with a generic complex scalar field theory on the lattice.
For simplicity, we always set the lattice spacing to be unity. Let $d$ be space-time dimension.
For $x\in\Integer^d$, we denote its components by
\begin{align}
x=(x_0,x_1,\dots,x_{d-1}),\quad x_\mu\in\Integer,\,\mu=0,1,\dots,d-1.
\end{align}
The $0$-th direction is called the time direction while
$k$-th directions with $k=1,2,\dots,d-1$ are called spacial directions.

For each $x\in\Integer^d$, there is a field algebra $\mathfrak{A}_x$, all the
polynomials generated by 
lattice fields $\{\phi(x),\phi(x)^*,1\}$. Let 
$\Lambda\subset\Integer^d$ be a finite subset of $\Integer^d$, and define a tensor product
\begin{align}
\mathfrak{A}_\Lambda:=\bigotimes_{x\in\Lambda}\mathfrak{A}_x,
\end{align}
which is all the polynomials of fields living in $\Lambda$, $\{ \phi(x),\phi(x)^*,1\}_{x\in\Lambda}$. 
The local field algebra $\mathfrak{A}_L$
is defined as
\begin{align}
\mathfrak{A}_L:=\bigcup_{\Lambda\subset\subset\Integer}\mathfrak{A}_\Lambda,
\end{align}
where $\Lambda\subset\subset\Integer$ means $\Lambda$ is a finite subset of $\Integer^d$.
The translation group $\Integer^d$ naturally acts on $\mathfrak{A}_L$, as automorphisms on $\mathfrak{A}_L$
\begin{align}
\Integer^d\ni y \mapsto \tau_y \in \mathrm{Aut}\, \mathfrak{A}_L , 
\end{align}
which is given by, for generators of $\mathfrak{A}_L$,
\begin{align}
\tau_y \phi(x)=\phi(x+y),\quad \tau_y\phi(x)^*=\phi(x+y)^*,\quad \tau_y(1)=1.
\end{align}
This clearly satisfies
\begin{align}
\tau_y \mathfrak{A}_x &= \mathfrak{A}_{x+y} \no\\
\tau_y \mathfrak{A}_\Lambda &= \mathfrak{A}_{\Lambda+y}.
\end{align}
The expectation value is a linear map
\begin{align}
\Expect{\cdot} :\mathfrak{A}_L \to \Complex,
\end{align}
satisfying the normalization condition
\begin{align}
\Expect 1 =1.
\end{align}
We call $\Expect{a}$ ``the expectation value of $a$" for $a\in\mathfrak{A}_L$. A lattice field theory is 
characterized by the pair $(\mathfrak{A}_L,\Expect{\cdot})$.

Next, we introduce operations on $\mathfrak{A}_L$, the reflection operator $\theta_l$ and $\theta_s$.
Consider the following two types of time reflections:
\begin{align}
(t,\bvec x)&\mapsto (-t+1, \bvec x),\\
(t,\bvec x)&\mapsto (-t, \bvec x).
\end{align}
The former is called link reflection and the latter site reflection. Clearly, link reflection is reflection 
with respect to the hyper plane $\{(t,\bvec x) \in\Integer^d\,;\, t=1/2\}$, and site reflection is to 
$\{(t,\bvec x) \in\Integer^d\,;\, t=0\}$. Corresponding to these two reflections, we define 
two operators $\theta_l$ and $\theta_s$ as follows. For the generators,
we define 
\begin{align}
\theta_l\phi(t,\bvec x)&:=\phi(-t+1,\bvec x)^*, \label{link ref}\\
\theta_s\phi(t,\bvec x)&:=\phi(-t,\bvec x)^* ,\label{site ref}
\end{align}
and $\theta_l1=\theta_s1=1$.
For general elements of $\mathfrak{A}_L$, we extend it by the relations
\begin{align}
\theta_\#(\alpha a + \beta b)&=\bar{\alpha}\theta_\#(a)+\bar{\beta}\theta_\#(b),\\
\theta_\#(ab)&=\theta_\#(b)\theta_\#(a),
\end{align}
for $\alpha,\beta\in\Complex$, $a,b\in\mathfrak{A}_L$ and $\#=l$ or $\#=s$. 
$\bar\alpha$ denotes complex conjugate of $\alpha \in \Complex$. 
Let $\mathfrak{A}^\#_\pm\subset\mathfrak{A}_L$ be a subalgebra which consists of all the polynomials 
generated by ``positive (resp. negative) time" field generators $\{\phi(t,\bvec x), \phi(t,\bvec x)^* \}$ and $1$. For
instance, $\mathfrak{A}_+^l$ is defined as
\begin{align}
\mathfrak{A}_+^l := \bigcup_{\Lambda\subset\subset\Integer^d_+}\mathfrak{A}_\Lambda,\quad
\Integer^d_+:=\{(t,\bvec x)\in\Integer^d\,:\, t\ge 1 \},
\end{align}
and similarly for other cases.
It is clear by the definition that $\theta_\#^2=1$ and $\theta_\#\mathfrak{A}_\pm^\#=\mathfrak{A}_\mp^\#$ for
$\#=l,s$. Since we mainly consider the link reflection in the following
discussion, we omit the subscript $l$
to mean ``link" , and just write 
$\theta:=\theta_l$, $\mathfrak{A}_+:=\mathfrak{A}_+^l$ and so forth. 

In the subsequent analyses, we consider lattice field theories 
$(\mathfrak{A}_L,\Expect\cdot)$
which fulfill the following assumptions \textbf{(A1)} - \textbf{(A4)}.
\begin{itemize}
\item[\textbf{(A1)}] Hermiticity :
\begin{align}
\Expect{\theta(a)}=\overline{\Expect{a}},\quad \forall a\in\mathfrak{A}_L.
\end{align}
\item[\textbf{(A2)}] Translational invariance of the expectation value :
\begin{align}
\Expect{\tau_x(a)}=\Expect{a},\quad \forall a\in\mathfrak{A}_L,\,\forall x\in\Integer^d.
\end{align} 
\item[\textbf{(A3)}] The link reflection positivity condition :
\begin{align}
\Expect{\theta(a)\,a} \ge 0 ,\quad \forall a\in\mathfrak{A}_+.
\end{align}
\item[\textbf{(A4)}] Polynomial boundedness of the correlation functions : For all $a,b\in\mathfrak{A}_L$, 
there exists some constant $C_{a,b}>0$ and $n\in\Natural$ such that
\begin{align}
|\Expect {b\,\tau_x (a)}|\le C_{a,b} (1+|x|^n),\quad x\in\Integer^d.
\end{align}
\end{itemize} 

If we assume in addition the following site reflection positivity condition, more strong result will be 
obtained :
\begin{itemize}
\item[\textbf{(A3S)}]  The site reflection positivity condition :
\begin{align}
\Expect{\theta_s(a)\,a} \ge 0 ,\quad \forall a\in\mathfrak{A}_+^s.
\end{align}
\end{itemize}
But, we will not assume this condition \textbf{(A3S)} in general except in subsection \ref{site}.
We stress that in order to reconstruct a Hilbert space of state vectors, and Hamiltonian and momentum
operators, \textbf{(A3S)} is not necessary. 

From the above assumptions \textbf{(A1)}-\textbf{(A4)},
a quantum mechanical system with a Hamiltonian and spacial momentum operators
can be reconstructed. 
Furthermore, our assumptions \textbf{(A1)}-\textbf{(A4)} are sufficient to
ensure that the two point function permits Euclidean version of Umezawa-Kamefuchi-K\"all\'en-Lehmann representation
with positive spectral density. More strictly, we can prove :

\begin{Thm}\label{main thm}
Suppose a lattice theory $(\mathfrak{A}_L,\Expect\cdot)$ satisfies
 \textbf{(A1)}-\textbf{(A4)}.  Then,
its Euclidean two point Green function $\Expect{\phi(x)^*\phi(0)}$ with $x_0=2m+1$,
$m=0,1,2,\dots$ is a Fourier (or 
Laplace, in the time direction) transformation
of some bounded Borel measure supported on $[0,\infty)\times [-\pi,\pi]^{d-1}$,
that is, there exists a $d$-dimensional bounded Borel measure $\rho$ with $\supp\rho\subset
[0,\infty)\times [-\pi,\pi]^{d-1}$ such that 
\begin{align}\label{eq: KL rep of prop}
\Expect{\phi(x)^*\phi(0)}=\int_{[0,\infty)\times [-\pi,\pi]^{d-1}} \,e^{-2m\lambda}{\rm e}^{-i\bvec p\cdot\bvec x}\,
d\rho(\lambda,\bvec p), \quad x_0=2m+1,\quad m=0,1,2,\dots.
\end{align}
Furthermore, the measure $\rho$ is unique in the sense that if there is a bounded 
Borel measure $\sigma$ satisfying the same relation as \eqref{eq: KL rep of prop}, then
$\rho=\sigma$. 
\end{Thm}

After we complete the proof of Theorem \ref{main thm},
it will become clear from the construction of the measure $\rho$ that it carries information of
the spectrum of quantum mechanical energy momentum operators $(H,\bvec P)$, and
we regard the expression of \eqref{eq: KL rep of prop} as the Euclidean version
of Umezawa-Kamefuchi-K\"all\'en-Lehmann representation of propagators. 
Apply the Lebesgue decomposition theorem to $\rho$ to obtain the decomposition
\begin{align}
\rho=\rho_a+\rho_s,
\end{align}
with $\rho_a$ absolutely continuous with respect to the Lebesgue measure $d\lambda\,d\bvec p$,
and $\rho_s$ singular to it. Let $\sigma$ be the Radon-Nikodym derivative of
$\rho_a$ :
\begin{align}
d\rho_a(\lambda,\bvec p)=\sigma(\lambda,\bvec p)\,\frac{d\lambda}{\pi}\frac{d\bvec p}{(2\pi)^{d-1}}.
\end{align}
We call $\sigma$ a spectral density function. 
By Theorem \ref{main thm}, if a lattice model 
$(\mathfrak{A}_L,\Expect\cdot)$ satisfies assumptions \textbf{(A1)} - \textbf{(A4)}, then
the spectral density of the model $\sigma$ has to be nonnegative at almost every $(\lambda,\bvec p)$
with respect to the Lebesgue measure.

\section{Construction of Hilbert space, Hamiltonian and Momentum Operators}\label{sec: lattice QM}
Before going to the proof of Theorem \ref{main thm},
we will review in detail how to reconstruct a quantum theory in a self-contained manner.
The discussion given in this section is mainly based on Refs. \cite{Ezawa1988xx,Seiler:1982pw,Osterwalder:1977pc}.

\subsection{Hilbert space of state vectors}
Let a lattice model $(\mathfrak{A}_L,\Expect\cdot)$ satisfy
the assumptions \textbf{(A1)}-\textbf{(A4)} 
given above. 
We emphasize that the site reflection positivity condition \textbf{(A3S)} is not assumed
here.
The reconstruction of a quantum theory can be performed 
without relying on the site reflection positivity \textbf{(A3S)}.
Remember again the subsprict $l$ for ``link" is omitted, for instance, 
$\theta:=\theta_l$, $\mathfrak{A}_+:=\mathfrak{A}_+^l$, and so forth. 

We will construct a Hilbert space of state vectors as follows.
Let us define a quadratic form on $\mathfrak{A}_+$
\begin{align}
(\cdot,\cdot)_+:\mathfrak{A}_+\times \mathfrak{A}_+\to \Complex
\end{align}
by
\begin{align}
(a,b)_+:=\Expect{\theta(a)b}\quad a,b\in\mathfrak{A}_+.
\end{align}
By virtue of the link reflection positivity \textbf{(A3)}, $(\cdot,\cdot)_+$ defines
a positive semi-definite inner product on $\mathfrak{A}_+$. Then, we 
consider 
the quotient vector space $\mathfrak{A}_+/\mathcal{N}$, where
\begin{align}
\mathcal{N}:=\{ a\in\mathfrak{A}_+ \, ;\, (a,a)_+=0 \}
\end{align}  
is the subspace of null vectors. Let us denote the equivalent class of $a\in\mathfrak{A}_+$
by $[a]$. The linear operation in $\mathfrak{A}_+/\mathcal{N}$
is given by
\begin{align}
[a]+[b]&:=[a+b] \quad a,b\in\mathfrak{A}_+,\\
\alpha[a]&:=[\alpha a] \quad \alpha\in\Complex,\, a\in\mathfrak{A}_+,
\end{align}
and the inner product $\langle\cdot,\cdot\rangle$ on $\mathfrak{A}_+/\mathcal{N}$ is defined by
\begin{align}
\Expect{[a],[b]}:=(a,b)_+\quad a,b\in\mathfrak{A}_+.
\end{align}
It is straightforward to check that these are
well-defined and that with this inner product $\mathfrak{A}_+/\mathcal{N}$
becomes a pre-Hilbert space (i.e. a complex vector space with positive definite inner product).
We define a Hilbert space $\mathcal{K}$ 
to be the completion of $\mathfrak{A}_+/\mathcal{N}$. Note that
the original vector space $\mathcal{D}:=\mathfrak{A}_+/\mathcal{N}$ is 
embedded as a dense subspace in $\mathcal{K}$. In general, $\mathcal{K}$ may be too large, 
containing unphysical states with infinite energy.
The physical Hilbert space is a closed subspace of $\mathcal{K}$ consisting of all the state
vectors with finite energy, which will be defined
after we introduce the transfer matrix $T$.

\subsection{Hamiltonian and Momentum operators}\label{Hamiltonian and momentum ops}
We will define translation operators $U_\mu$ for
$\mu=0,1,2,\dots,d-1$ on $\mathcal{K}$, through the translation automorphisms $\{\tau_y\}_{y\in\Integer^d}$.
For elements of $\mathcal{D}$, we define $U_\mu$ by the relations 
\begin{align}\label{eq: definition of U_mu}
U_\mu[a]:=[\tau_{\mu} (a)],\quad a\in\mathfrak{A}_+,\,\mu=0,1,2,\dots,d-1,
\end{align}
where $\tau_\mu$ is the one-site translation in the $\mu$-th direction:
\begin{align}
\tau_\mu:=\tau_{e_\mu},
\end{align} 
with $e_\mu$ being a unit vector in the $\mu$-th direction.
For spacial directions $k=1,2,\dots,d-1$, $U_k$ is the unitary implementation of spacial
translations $\{\tau_{\bvec x}\}_{\bvec x\in\Integer^{d-1}}$ in $\mathcal{K}$. However, in the time
direction $\mu=0$, the ``$\theta$-reflected" inner product in $\mathcal{K}$ makes $U_0$ 
a self-adjoint operator instead of a unitary operator. This is because we are
working in the Euclidean space-time.

Let us show the operation of $U_\mu$ does not
depend on the choice of representatives, i.e. $\tau_\mu \mathcal{N}\subset\mathcal{N}$,
so that $U_\mu$'s are well-defined.

First, consider the spacial directions. Let $||a||_+$ be 
the square root of $(a,a)_+$ with $a\in\mathfrak{A}_+$.
Since $\tau_k$'s, $(k=1,2,\dots,d-1)$ commute with $\theta$ and $\Expect\cdot$
is translationally invariant by \textbf{(A2)}, we learn
\begin{align}
(\tau_k (a), \tau_k (b))_+=(a,b)_+,
\end{align}
implying 
\begin{align}\label{iso of tau}
||\tau_k (a)||_+=||a||_+,
\end{align}
and therefore $\tau_k\mathcal{N}\subset{N}$. This confirms
the well-definedness of $U_k$'s $(k=1,2,\dots,d-1)$, and we learn from \eqref{iso of tau}
\begin{align}
|| U_k[a]||_\mathcal{K}=|| [a] ||_\mathcal{K}\quad a\in\mathfrak{A}_+.
\end{align}
Thus, we conclude that the operator norm\footnote{
For a linear operator $A$ in a Hilbert space, its operator norm is defined
as
\[ ||A||:=\sup_{x\in D(A)}\frac{||Ax||}{||x||}, \]
where $D(A)$ is a domain of definition of $A$.
} of $U_k$, which we denote by $||U_k||$, 
satisfies $||U_k||=1$ and then $U_k$ has the unique unitary extension on $\mathcal{K}$. 
We denote this extended unitary operator by the same symbol $U_k$.

Next, we consider the time direction. This case needs more arguments \cite{Ezawa1988xx}. 
To see that $U_0$ is well-defined, 
put for $t=0,1,2,\dots$ and for some fixed $a\in\mathfrak{A}_+$,
\begin{align}
F(t):=(a,\tau_0^t (a))_+.
\end{align}
Noting that $\tau_0\circ\theta=\theta\circ \tau_0^{-1}$ on $\mathfrak{A}_+$ and using the translational
invariance of $\Expect\cdot$ \textbf{(A2)}, we obtain
\begin{align}\label{eq: self-adjointness of T_0}
(a,\tau_0 (b))_+&=\Expect{\theta(a)\,\tau_0(b)}\no\\
	&=\Expect{\tau_0^{-1}\circ\theta(a)\,(b)}\no\\
	&=\Expect{\theta\circ \tau_0(a)\,(b)}\no\\
	&=(\tau_0 (a), b)_+.
\end{align}
From the link reflection positivity condition \textbf{(A3)} and \eqref{eq: self-adjointness of T_0}, we obtain
\begin{align}\label{F positive even}
F(2t)=(a,\tau_0^{2t}(a))_+=||\tau_0^t (a) ||_+^2\ge 0,
\end{align}
and by the repeated use of \eqref{eq: self-adjointness of T_0} and  Schwarz's inequality for $(\cdot,\cdot)_+$, we also have
\begin{align}\label{eq: estimation of F(t)}
0\le F(2t)&\le || a ||_+ || \tau_0^{2t}(a) ||_+ \no\\
	&=|| a ||_+ (\tau_0^{2t} (a),\tau_0^{2t}(a))_+^{1/2} \no\\
	&=|| a ||_+ ( a,\tau_0^{2^2t}(a))_+^{1/2}\no\\
	&\le || a ||_+^{1+1/2} || \tau_0^{2^2t} (a) ||_+^{1/2}\no\\
	&=\dots \no\\
	&\le || a ||_+^{1+1/2 + (1/2)^2+\dots+(1/2)^n} ||\tau_0^{2^{n+1} t} (a)||_+^{(1/2)^n}\no\\
	&=|| a ||_+^{1+1/2 + (1/2)^2+\dots+(1/2)^n}F(2^{n+2}t)^{(1/2)^{n+1}}.
\end{align}
By the polynomial boundedness \textbf{(A4)}, there exists some constant $C_a>0$ and $m\in\Natural$
such that,
\begin{align}
F(t)\le C_a(1+|t|^m).
\end{align}
Hence, noting that
\begin{align}
\lim_{n\to\infty} F(2^{n+2}t)^{(1/2)^{n+1}}\le \lim_{n\to\infty}\{C_a(1+|2^{n+2}t|^m)\}^{(1/2)^{n+1}}=1,
\end{align}
and taking the limit $n\to\infty$ in \eqref{eq: estimation of F(t)}, we obtain
the estimation
\begin{align}
0\le F(2t) \le || a ||_+^2.
\end{align} 
In particular, if $a\in\mathcal{N}$, we learn
\begin{align}
|| \tau_0 (a) ||_+^2 = F(2)\le || a ||_+^2 =0,
\end{align}
showing that $\tau_0\mathcal{N}\subset\mathcal{N}$. Then, it is confirmed that
the definition of $U_0$ by 
\eqref{eq: definition of U_mu} makes sense. Further, one can see that $U_0$ is bounded because
the computation 
\begin{align}
|| U_0[a] ||_{\mathcal{K}} = || [\tau_0 (a)] ||_\mathcal{K} = || \tau_0(a) ||_+ \le || a ||_+ = || [a] ||_\mathcal{K}
\end{align}
shows that the operator norm of $U_0$ is less than $1$ :
\begin{align}
|| U_0 || \le 1.
\end{align}
Therefore, from the boundedness of $U_0$ and \eqref{eq: self-adjointness of T_0}, 
$U_0$ has the unique self-adjoint extension $T$ defined on $D(T)=\mathcal{K}$.
$T$ is called the transfer matrix.
Here, the operator domain of $A$ is denoted by $D(A)$.

We will construct the Hamiltonian and momentum operators from $T$ and $U_k$'s.
The desired relation between Hamiltonian and the transfer matrix $T$ is 
\begin{align}\label{one step trans}
T^t={\rm e}^{-tH}, \,t=0,1,2,\dots,
\end{align}
but this can not be satisfied in general, because $T$ is not always a nonnegative operator
and may have non-trivial kernel.
So, we proceed in the following way. First, consider
\begin{align}
\mathcal{H}:=(\ker T)^\perp.
\end{align}
We do not want state vectors in $\ker T$ to be contained in the physical Hilbert space
since these states possess infinite energy. Hence, it is reasonable to regard the Hilbert space $\mathcal{H}$
as a physical Hilbert space. Next, define Hamiltonian $H$ by
\begin{align}
H:=-\frac{1}{2}\ln (T|_\mathcal{H})^2=-\ln |\,T|_\mathcal{H}\, |.
\end{align}
From the functional calculus (see, for example, Theorem VIII.5 in Ref. \cite{Reed:1980}), 
$H$ is a densely defined, positive self-adjoint operator in $\mathcal{H}$.
From now on, we will denote the transfer matrix $T|_\mathcal{H}$ just by $T$.
Note that $H$ is unbounded in general, and is bounded if
and only if $0\not\in\sigma(T)$ and also that since $\sigma(T)\subset[0,1]$,
the spectrum of $H$ is contained in $[0,\infty)$. The important remark here is that
it is true that
\begin{align}
T^{t}=\e^{-t H},\quad t=0,2,4,\dots,
\end{align}
but \eqref{one step trans} is \textit{false} in general for odd $t$.

Spacial momentum operators $P_k\,(k=1,2,\dots,d-1)$'s are defined through the expected relations :
\begin{align}
U_k^{n}={\rm e}^{-in P_k}, \quad k=1,2,\dots,d-1,\quad n\in\Integer.
\end{align}
Here, $U_k$'s are considered to be unitary operators in $\mathcal{H}$.
Explicitly, we define $P_k$ as follows.
Let $\mathcal{B}^d$ be the $d$ dimensional 
Borel field. By the spectral theorem for the unitary operators, there exist a unique
one dimensional spectral measure $\{E_k(B)\}_{B\in\mathcal{B}^1}$ supported on $[-\pi,\pi]$
such that
\begin{align}
U_k = \int_{[-\pi,\pi]} {\rm e}^{-i\theta}\,dE_k(\theta).
\end{align}
Define
\begin{align}
P_k:=\int_\Real \theta\,dE_k(\theta)=\int_{[-\pi,\pi]} \theta\,dE_k(\theta),
\end{align}
and this $P_k$ is a bounded self adjoint operator in $\mathcal{H}$ satisfying
\begin{align}
{\rm e}^{-in P_k}=\int_{[-\pi,\pi]} {\rm e}^{-in \theta}\,dE_k(\theta)=U_k^{n},
\end{align}
for all $n\in\Integer$.

So far, we have constructed the energy-momentum operators $(H,\bvec P)$ and the
physical energy-momentum spectrum is considered to be the joint spectrum of
$(H,\bvec P)$.  But, the existence of the joint spectrum is ensured only when
they are strongly commuting (i.e. their associated spectral measures are commuting). 
Thus, we have to prove :

\begin{Thm}\label{strong com} 
Let $(H,\bvec P)$, with $\bvec P=(P_1,P_2,\dots,P_{d-1})$ be Hamiltonian and 
momentum operators defined above. Then,
$(H,\bvec P)$ are strongly commuting, that is, all the spectral projections are commuting with each other.
\end{Thm}
\begin{proof}
First, we show the strong commutativity of $S:=T^2$ and $P_k$ ($k=1,2,\dots,d-1$).
Because $S$ and $P_k$ are bounded, it suffices to prove they commute
in the ordinary sense. Take arbitrary $u,v\in\mathcal{H}$.
Let $f\in C([-\pi,\pi])\cap C^1((-\pi,\pi))$ with $f(-\pi)=f(\pi)$, $f'\in L^2([-\pi,\pi])$ and $c_n$ be its Fourier 
coefficients :
\begin{align}
c_n:=\frac{1}{2\pi}\int_{[-\pi,\pi]} f(\lambda){\rm e}^{in\lambda}\,d\lambda.
\end{align} 
Then, $\{c_n\}_n\in l^2(\Integer)$ and $f$ can be written as an
infinite series
converging absolutely and uniformly in $\lambda$ \cite{Kuroda:1980}:
\begin{align}
f(\lambda)=\sum_{n\in\Integer}c_n {\rm e}^{-in\lambda}.
\end{align}
Since $S$ and $U_k^n$ are commuting for all $n\in\Integer$, we obtain for all $f\in C([-\pi,\pi])\cap C^1((-\pi,\pi))$ 
with $f(-\pi)=f(\pi)$ and $f'\in L^2([-\pi,\pi])$,
\begin{align}
\Expect{u, f(P_k)S\,v}
&=\int_{[-\pi,\pi]}f(\lambda) \,d\Expect{u,E_k(\lambda)S\,v}\no\\
&=\int_{[-\pi,\pi]}
\sum_{n\in\Integer}c_n{\rm e}^{-in\lambda} \,d\Expect{u,E_k(\lambda)S\,v}\no\\
&=
\sum_{n\in\Integer}c_n\int_{[-\pi,\pi]} {\rm e}^{-in\lambda} \,d\Expect{u,E_k(\lambda)S\,v}\no\\
&=\sum_{n\in\Integer}c_n\Expect{u,{\rm e}^{-inP_k}S\,v}\no\\
&=\sum_{n\in\Integer}c_n\Expect{u,U_k^n S\,v}\no\\
&=\sum_{n\in\Integer}c_n\Expect{u,S U_k^n \,v}\no\\
&=\sum_{n\in\Integer}c_n\int_{[-\pi,\pi]} {\rm e}^{-in\lambda}\,
d\Expect{u,S E_k(\lambda) \,v}\no\\
&=\int_{[-\pi,\pi]}\sum_{n\in\Integer}c_n{\rm e}^{-in\lambda}\,
d\Expect{u,S E_k(\lambda) \,v}\no\\
&=\int_{[-\pi,\pi]} f(\lambda)\,
d\Expect{u,S E_k(\lambda) \,v}\no\\
&=\Expect{u,f(P_k)S\,v},
\end{align}
by Fubini's theorem. Hence, for arbitrary $f\in C([-\pi,\pi])\cap C^1((-\pi,\pi))$ with $f(-\pi)=f(\pi)$, 
$f'\in L^2([-\pi,\pi])$, we have
\begin{align}\label{com of FT}
f(P_k)S=Sf(P_k).
\end{align}
Choose
a sequence $\{f_m\}_{m=1}^\infty$ such that 
\begin{align}
f_m\in C([-\pi,\pi])\cap C^1((-\pi,\pi)),\quad f_m(-\pi)=f_m(\pi),\,f_m'\in L^2([-\pi,\pi]),\,m=1,2,\dots,
\end{align}
and
\begin{align}
\lim_{m\to\infty}f_m(\lambda)=\lambda, \quad d\lambda - \text{a.e. } .
\end{align}
For instance, one may adopt
\begin{align}
f_m(\lambda)=\begin{cases}
\lambda-\pi\left(\frac{\lambda}{\pi}\right)^m & 0\le \lambda \le \pi\\
\\
\lambda+\pi\left(-\frac{\lambda}{\pi}\right)^m & -\pi \le \lambda < 0.
\end{cases}
\end{align}
By the Lebesgue convergence theorem, we obtain for all $u,v\in\mathcal{H}$, 
\begin{align}
\Expect{u,  P_kS\,v}&=\int_{[-\pi,\pi]} \lambda\, d\Expect{u, E_j(\lambda)S v} \no\\
	&=\int_{[-\pi,\pi]} \lim_{m\to\infty}f_m(\lambda)\, d\Expect{u, E_j(\lambda)S v} \no\\
	&=\lim_{m\to\infty}\int_{[-\pi,\pi]} f_m(\lambda)\, d\Expect{u, E_j(\lambda)S v} \no\\
	&=\lim_{m\to\infty}\Expect{u,f_m(P_k)S v}\no\\
	&=\lim_{m\to\infty}\Expect{u,S f_m(P_k) v}\no\\
	&=\lim_{m\to\infty}\int_{[-\pi,\pi]} f_m(\lambda)\, d\Expect{u, SE_j(\lambda) v} \no\\
	&=\int_{[-\pi,\pi]} \lim_{k\to\infty}f_m(\lambda)\, d\Expect{u, SE_j(\lambda) v} \no\\
	&=\int_{[-\pi,\pi]} \lambda\, d\Expect{u, SE_j(\lambda) v} \no\\
	&=\Expect{u,S P_k v}.
\end{align}
which results in
\begin{align}
P_k S=SP_k, \quad k=1,2,\dots,d-1.
\end{align}

Next, in order to prove the strong commutativity of $H$ and $P_k$, take any
real-valued Borel function $F$ satisfying $E_S(\{\lambda\in\Real\,:\,|F(\lambda)|=\infty\})=0$.
$F(S)$ is a (possibly unbounded) self-adjoint operator. Suppose we can show
\begin{align}\label{half com}
P_kF(S)\subset F(S)P_k.
\end{align}
Then, the bijectivity
of $F(S)-z$ and $P_k-w$, for arbitrary $z,w\in \Complex\setminus\Real$, shows that 
\begin{align}
(F(S)-z)^{-1}(P_k-w)^{-1}v=(P_k-w)^{-1}(F(S)-z)^{-1}v,\quad \forall v\in\mathcal{H}.
\end{align}
Therefore, $F(S)$ and $P_k$ are strongly commuting. Choosing 
\begin{align}
F(\lambda)=-\frac{1}{2}\ln \lambda,
\end{align}
so that $F(S)=H$ proves the strong commutativity of $H$ and $P_k$.

It remains to show \eqref{half com}. Suppose $v\in D(F(S))$. Then, from the strong
 commutativity of $S$ and $P_k$, we learn
\begin{align}
\int_\Real |F(\lambda)|^2\,d\Expect{P_k v,E_S(\lambda)P_k v}
&=\int_\Real |F(\lambda)|^2\,d\Expect{P_k^2v,E_S(\lambda)v}\no\\
&\le ||P_kv||^2\,\int_\Real |F(\lambda)|^2\,d|| E_S(\lambda)v||^2\no\\
&<\infty.
\end{align}
This shows that $v\in D(F(S))$ implies $P_k v\in D(F(S))$. Furthermore, for all 
$u\in\mathcal{H}$ and $v\in D(F(S))$,
we obtain
\begin{align}
\Expect{u,F(S)P_k v}&=\int F(\lambda)\,d\Expect{u,E_S(\lambda)P_k v}\no\\
&=\int F(\lambda)\,d\Expect{u,P_k E_S(\lambda) v}\no\\
&=\Expect{u,P_k F(S) v},
\end{align}
which shows \eqref{half com}.

The strong commutativity of $P_j$ and $P_k$ is similar and easier.
From the commutativity of $U_j$ and $U_k$, and the Fubini's theorem,
we have for any $f\in C([-\pi,\pi])\cap C^1((-\pi,\pi))$ with $f(-\pi)=f(\pi)$, $f'\in L^2([-\pi,\pi])$,
and for all $u,v\in\mathcal{H}$,
\begin{align}
\Expect{u, f(P_j) f(P_k) v}
&=\sum_{n,m\in\Integer}c_n c_m \Expect{u,U_j^n U_k^m v}\no\\
&=\sum_{n,m\in\Integer}c_n c_m\Expect{u,U_k^m U_j^n v}\no\\
&=\Expect{u,f(P_k) f(P_j) v}.
\end{align}
This shows 
\begin{align}
 f(P_j)  f(P_k)=f(P_k) f(P_j),
\end{align}
for all $f\in C([-\pi,\pi])\cap C^1([-\pi,\pi])$ with $f(-\pi)=f(\pi)$ and $f'\in L^2([-\pi,\pi])$. By the same limiting argument as above,
we obtain
\begin{align}
P_j P_k=P_k P_j,
\end{align}
completing the proof.
\end{proof}

Let $\{E_0(\cdot)\}$ the spectral measure of $H$, and $\{E_k(\cdot)\}$ be that of
$P_k$.
From Theorem \ref{strong com}, we can define the product spectral measure on $\Real^d$
\begin{align}
E:=E_0\times E_1\times E_2 \times\dots\times E_{d-1},
\end{align}
and the joint spectrum
\begin{align}
\sigma_J(E):=\supp E.
\end{align}
It is clear by the definition that
\begin{align}
\sigma_J(E)&\subset \sigma(H) \times \sigma(P_1)\times\dots \times \sigma(P_{d-1})  
\subset [0,\infty)\times [-\pi,\pi]^{d-1}.
\end{align}
In fact, this follows from 
\begin{align}
E(\sigma(H) \times \sigma(P_1)\times\dots \times \sigma(P_{d-1}) )
=E_0(\sigma(H))E_1(\sigma(P_1))\dots E_{d-1}(\sigma(P_{d-1})),
\end{align}
and $\sigma (H)= \supp E_0$, $\sigma(P_k)=\supp E_k$, $k=1,2,\dots,d-1$, and 
the definition of the support. $(H,\bvec P)$ can be expressed in terms of 
the product spectral measure $\{E(\cdot)\}$ by
\begin{align}
H&=\int_{\Real^d} \lambda \,dE(\lambda,\bvec p)=\int_{[0,\infty)\times[-\pi,\pi]^{d-1}}
 \lambda \,dE(\lambda,\bvec p),\no\\
P_k&=\int_{\Real^d} p_k \,dE(\lambda,\bvec p)=\int_{[0,\infty)\times[-\pi,\pi]^{d-1}} p_k \,dE(\lambda,\bvec p).
\end{align}
By these expressions, we learn
\begin{align}\label{FC joint}
\e^{-tH}{\rm e}^{-iP_1x_1}\dots{\rm e}^{-iP_{d-1}x_{d-1}}&=\e^{\overline{-tH-i\bvec P\cdot\bvec x}}
=\e^{-tH-i\bvec P\cdot\bvec x}\no\\
&=\int_{[0,\infty)\times[-\pi,\pi]^{d-1}} \e^{-\lambda t-i\bvec p\cdot\bvec x} \,dE(\lambda,\bvec p),
\end{align}
for $t=0,1,2,\dots$ and $\forall \bvec x\in\Integer^{d-1}$. For a closable operator $A$, $\bar A$
is its closure. In the third equality above, we have used the fact that $-tH-i\bvec P\cdot\bvec x$
is closed since $\bvec P\cdot\bvec x$ is bounded. 
We employ in what follows the notation
\begin{align}
U^{\bvec x}:={\rm e}^{-i\bvec P\cdot\bvec x}, \quad \bvec x\in\Integer^{d-1}.
\end{align}

\subsection{Remark: Site reflection positive case}\label{site}
In this subsection, the site reflection positivity condition \textbf{(A3S)} is assumed
in addition. In this case, $T\ge 0$ can be shown \cite{Ezawa1988xx}. 

Let $\theta_s$ be the site reflection defined by \eqref{site ref}. Note that the site reflection
$\theta_s$ is related to the link reflection $\theta$ by
\begin{align}
\theta=\tau_0\circ\theta_s.
\end{align}
Then, our additional assumption \textbf{(A3S)} can be read as
\begin{align}
\Expect{\theta_s(a)\,a}=\Expect{\tau_0^{-1}\circ\theta(a)\, a}=\Expect{\theta(a)\,\tau_0(a)}\ge0,
\end{align} 
for all $a\in\mathfrak{A}_+$.

Define as before for $t=0,1,2,\dots$ and for some fixed $a\in\mathfrak{A}_+$,
\begin{align}
F(t):=(a,\tau_0^t (a))_+.
\end{align}
Under the present assumptions \textbf{(A1)}-\textbf{(A4)} and \textbf{(A3S)},
 it can also be shown that not only \eqref{F positive even}
but also the value of $F$ at odd integers is also positive :
\begin{align}\label{pos at odd}
F(2t+1)&=(a,\tau_0^{2t+1} (a))_+ \no\\
	&=(\tau_0^{t}(a),\tau_0\circ\tau_0^{t}(a))_+\no\\
	&=\Expect{\theta(\tau_0^t (a))\, \tau_0(\tau_0^{t}(a))}\no\\
	&\ge 0,
\end{align}
for $t=0,1,2,\dots$. 
Here, we have used the site reflection positivity \textbf{(A3S)} and the translational invariance of
the expectation functional \textbf{(A2)}.
Combining \eqref{F positive even} and \eqref{pos at odd}
we obtain for \textit{all} $t=0,1,2,\dots$,
\begin{align}
F(t)\ge 0.
\end{align}
From this, it immediately follows that $T$ satisfies
\begin{align}
\Expect{[a],T^t[a]}_{\mathcal{H}}\ge 0,
\end{align}
for all $t=0,1,2\dots$, which implies 
\begin{align}
T \ge 0.
\end{align}
Hence, in the present case,
\begin{align}
T^t=\e^{-tH}
\end{align}
is true not only for even but also for odd $t\ge0$.

\section{Proof of Theorem \ref{main thm}}\label{sec: proof}
Now, we have collected sufficient materials to prove the Theorem \ref{main thm}.
After proving Theorem \ref{main thm}, we will discuss the Umezawa-Kamefuchi-K\"all\'en-Lehmann representation
of the propagator of a Euclidean lattice theory, and point out that 
a lattice bosonic model with negative spectral density must violates 
at least one of our assumptions \textbf{(A1)} - \textbf{(A4)}.
\subsection{Proof of Theorem \ref{main thm}}
The existence of such measure is a simple corollary of the above construction. In fact,
since
\begin{align}
\Expect{\phi(x)^*\phi(0)}&=\Expect{\phi(0)\phi(x)^*}\no\\
&=\Expect{\theta(\phi(1,\bvec 0)^*) \,\tau_{(x_0-1,\bvec x)}
(\phi(1,\bvec 0)^*)}\no\\
&=\Expect{[\phi(1,\bvec 0)^*], 
[\tau_{x_0-1,\bvec x)}(\phi(1,\bvec 0)^*)]}\no\\
&=\Expect{[\phi(1,\bvec 0)^*],
 T^{x_0-1}U^{\bvec x}
[\phi(1,\bvec 0)^*]},
\end{align}
and 
\begin{align}
T^{x_0-1}=T^{2m}=e^{-2mH},
\end{align}
one obtains by using \eqref{FC joint}
\begin{align}
\Expect{[\phi(1,\bvec 0)^*],
 T^{x_0-1}U^{\bvec x}
[\phi(1,\bvec 0)^*]}&=\Expect{[\phi(1,\bvec 0)^*],
 e^{-2mH}{\rm e}^{-i\bvec P\cdot \bvec  x}
[\phi(1,\bvec 0)^*]}\no\\
&=\int_{[0,\infty)\times [-\pi,\pi]^{d-1}}
e^{-2m\lambda}{\rm e}^{-i\bvec p\cdot\bvec  x}
\,d || E(\lambda,\bvec p)[\phi(1,\bvec 0)^*]||^2.
\end{align}
This proves the existence.

The uniqueness remains to be proved. Clearly,
it suffices to prove the following proposition.

\begin{Prop}\label{prop: uniqueness}
For an $\Real$-valued $d$-dimensional Borel measure
(i.e. bounded signed Borel measure) $\sigma$ supported on $K:=[0,\infty)\times[-\pi,\pi]^{d-1}$,
define
\begin{align}
\tilde\sigma(m,\bvec x):=\int_K {\rm e}^{-2m\lambda-i\bvec p\cdot\bvec x}\,d\sigma(\lambda,\bvec p), 
\quad m=0,1,\dots,\quad\bvec x\in\Integer^{d-1}.
\end{align}
Suppose 
that given two $\Real$-valued Borel measures $\sigma_1$ 
and $\sigma_2$ satisfy
\begin{align}
\tilde{\sigma}_1(m,\bvec x)=\tilde\sigma_2(m,\bvec x),
\end{align}
for all $m=0,1,\dots$ and $x\in\Integer^{d-1}$.
Then, $\sigma_1=\sigma_2$.
\end{Prop}

\begin{proof}
Let $A\subset [0,1]$ and $B\subset [-\pi,\pi]^{d-1}$ be Borel sets, and 
$f(\lambda)={\rm e}^{-2\lambda}$.
and define
\begin{align}
\mu_i(A\times B):=\sigma_i(f^{-1}(A)\times B),\quad i=1,2.
\end{align}
By Hopf's extension theorem, this relation defines bounded singed Borel measures $\mu_1,\mu_2$ on 
$[0,1]\times[-\pi,\pi]^{d-1}$. Since $\mu_1=\mu_2$ implies $\sigma_1=\sigma_2$,
to show $\sigma_1=\sigma_2$ is sufficient.

The assumption $\tilde\sigma_1=\tilde\sigma_2$ is
equivalent to
\begin{align}
\int_{[0,1]\times[-\pi,\pi]^{d-1}} t^m {\rm e}^{-i\bvec p\cdot \bvec x}\, d\mu_1(t,\bvec p)
=\int_{[0,1]\times[-\pi,\pi]^{d-1}} t^m {\rm e}^{-i\bvec p\cdot \bvec x}\, d\mu_2(t,\bvec p)
\end{align}
for all $m=0,1,\dots$ and $x\in\Integer^{d-1}$.
This implies that for all polynomial $P(t)$ on $[0,1]$,
\begin{align}\label{eq: pol and exp relation}
\int_{[0,1]\times[-\pi,\pi]^{d-1}} P(t) {\rm e}^{-i\bvec p\cdot \bvec x}\, d\mu_1(t,\bvec p)
=\int_{[0,1]\times[-\pi,\pi]^{d-1}} P(t) {\rm e}^{-i\bvec p\cdot \bvec x}\, d\mu_2(t,\bvec p).
\end{align}
Let $T^{d-1}$ be a $d-1$ dimensional torus, that is, $T^{d-1}=[-\pi,\pi]^{d-1}$ with 
$-\pi$ and $\pi$ identified. 
Take arbitrary $f\in C(X)$, a continuous function on
 $X:=[0,1]\times T^{d-1}$.
From \eqref{eq: pol and exp relation}, we obtain
\begin{align}
\int_X f\,d \mu_1 =\int_X f\,d\mu_2, \quad f\in C(X).
\end{align}
By the uniqueness statement of the Riesz-Markov theorem (\cite{Reed:1980}, Theorem IV.14),
we have
\begin{align}
\mu_1(A)=\mu_2(A),\quad A\in\mathcal{B}_X,
\end{align}
where $\mathcal{B}_X$ is the Borel field in $X$.
This implies
$\mu_1=\mu_2$ and then $\sigma_1=\sigma_2$,
which completes the proof.
\end{proof} 

\subsection{Umezawa-Kamefuchi-K\"all\'en-Lehmann representation in more general cases}
Suppose a lattice model $(\mathfrak{A}_L,\Expect\cdot)$ which may not satisfy
\textbf{(A1)} - \textbf{(A4)} happens to have the representation
\begin{align}\label{eq: general KL}
\Expect{\phi(x)^*\phi(0)}\Big|_{x_0>0}=\int_{[0,\infty)\times[-\pi,\pi]^{d-1}} 
{\rm e}^{-\lambda (x_0-1)-i\bvec p\cdot\bvec x}\,d\rho(\lambda,\bvec p).
\end{align}
with some bounded signed Borel measurer $\rho$, which does not have to be a positive measure
in this case. We also call \eqref{eq: general KL} the Umezawa-Kamefuchi-K\"all\'en-Lehmann representation
even if the lattice theory $(\mathfrak{A}_L,\Expect\cdot)$ violates some of our assumptions
\textbf{(A1)} - \textbf{(A4)}.
By Proposition \ref{prop: uniqueness}, this $\rho$ is uniquely 
determined by the above representation \eqref{eq: general KL}.
We also call a spectral density, the Radon-Nikodym derivative of the absolutely continuous part of $\rho$.

 Combining Theorem \ref{main thm} and Proposition \ref{prop: uniqueness},
one concludes that :
\begin{Thm}\label{cor: violation}
A lattice model $(\mathfrak{A}_L,\Expect\cdot)$ with a spectral density $\sigma$
which becomes negative on a set of positive Lebesgue measure,
breaks at least one of our assumptions \textbf{(A1)} - \textbf{(A4)}. 
\end{Thm}
\begin{proof}
By the hypotheses, it has  the Umezawa-Kamefuchi-K\"all\'en-Lehmann representation of propagators 
\eqref{eq: general KL} with some $\rho'$.
Put $x_0=2m+1\,(m=0,1,\dots)$ and then, we have
\begin{align}
\Expect{\phi(x)^*\phi(0)}=\int_{[0,\infty)\times[-\pi,\pi]^{d-1}} 
{\rm e}^{-2m\lambda-i\bvec p\cdot\bvec x}\,d\rho'(\lambda,\bvec p),\quad x_0=2m+1,\,m=0,1,\dots.
\end{align}
where $\rho'$ is some bounded singed Borel measure, and its spectral density we denote by $\sigma'$
is not a non-negative function by the hypotheses.

Suppose, toward a contradiction, the lattice model $(\mathfrak{A}_L,\Expect\cdot)$
satisfies all the assumptions \textbf{(A1)} - \textbf{(A4)}. Then, Theorem \ref{main thm} says that
it has a representation 
\begin{align}
\Expect{\phi(x)^*\phi(0)}=\int_{[0,\infty)\times[-\pi,\pi]^{d-1}} 
{\rm e}^{-2m\lambda-i\bvec p\cdot\bvec x}\,d\rho(\lambda,\bvec p),\quad x_0=2m+1,\,m=0,1,\dots,
\end{align}
with a spectral density $\sigma$, which is associated with $\rho$, and 
$\sigma$ has to be a non-negative function.

However, by Proposition \ref{prop: uniqueness}, $\rho$ and $\rho'$ must be identical,
and therefore $\sigma$ and $\sigma'$ are identical almost everywhere with respect
to the Lebesgue measure,
which is a contradiction.
\end{proof}

\section{Violation of the reflection positivity in Lattice overlap boson system}\label{sec: overlap boson}
We will discuss the lattice overlap boson system as an application of  the previous
results. It will be shown that the free overlap boson system violates the reflection positivity 
condition.

\subsection{Klein-Gordon type Overlap Dirac Operator}
The free overlap Dirac operator \cite{Neuberger:1997fp,Neuberger:1998wv},
\begin{align}
D:l^2(\Integer^d;\Complex^{2^{d/2}})\to l^2(\Integer^d;\Complex^{2^{d/2}}), 
\end{align}
is defined through the lattice Wilson Dirac operator,
which is bounded normal operator on $l^2(\Integer^d;\Complex^{2^{d/2}})$,
\begin{align}
D_w=\sum_{\mu=0}^{d-1}\left(\frac{1}{2}\gamma_\mu(\partial_\mu-\partial_\mu^\dagger)+
\frac{1}{2}\partial_\mu\partial_\mu^\dagger\right),
\end{align}
by
\begin{align}
D=\frac{1}{R}\left(1+\frac{D_w-1}{\sqrt{(D_w-1)^\dagger(D-1)}}\right), 
\end{align}
with a parameter $R>0$.
Here, for a linear operator $A$ in $l^2(\Integer^d,\Complex^{2^{d/2}})$, $A^\dagger$ denotes its adjoint
operator. $\partial_\mu$ is a forward differential operator on the lattice :
\begin{align}
(\partial_\mu u)(x):=u(x+e_\mu)-u(x),\quad u\in l^2(\Integer^d;\Complex^{2^{d/2}})
\end{align}
and $\partial_\mu^\dagger$ is a backward differential operator, which is the adjoint operator
of $\partial_\mu$, given by
\begin{align}
(\partial_\mu^\dagger u)(x):=u(x)-u(x+e_\mu),\quad u\in l^2(\Integer^d;\Complex^{2^{d/2}}).
\end{align} 
$\{\gamma_\mu\}_{\mu=0}^{d-1}$ are the Euclidean gamma matrices, satisfying
\begin{align}
\gamma_\mu\gamma_\nu+\gamma_\nu\gamma_\mu=2\delta_{\mu\nu}\hat{1},\quad \mu,\nu=0,1,\dots,d-1.
\end{align}
The physical motivation for the use of the overlap Dirac operator $D$ is that 
it satisfies the Ginsparg-Wilson relation \cite{Ginsparg:1981bj}
thanks to which the chiral symmetry can be suitably defined on the lattice \cite{Luscher:1998pqa}
in spite of the 
notorious no-go theorem of Neilsen and Ninomiya \cite{Nielsen:1980rz,Nielsen:1981xu}.

The lattice overlap boson system is characterized by a ``Klein-Gordon" type operator,
which we denote by $\square$,
defined in terms of the overlap Dirac operator by the relation :
\begin{align}
\square\cdot\hat{1}=D^\dagger D.
\end{align}
Note that $D^\dagger D$ no longer has a spinor structure and then 
proportional to unit matrix $\hat{1}$ in the spinor indices. As a result, 
the overlap boson operator $\square$ given above is considered to be an operator
in $l^2(\Integer^d)$. 

By the definition, $\square$ is a non-negative
self-adjoint linear operator. We will further analyze 
the properties of $\square$.
Let $\mathcal{F}$ be the Fourier transformation from the position space ($x$-space) 
$l^2(\Integer^d)$ 
to the momentum space ($p$-space)
$L^2([-\pi,\pi]^d, dp/(2\pi)^d)$ :
\begin{align}
\mathcal{F}:l^2(\Integer^d)\to L^2([-\pi,\pi]^d, dp/(2\pi)^d),
\end{align}
defined by
\begin{align}
(\mathcal{F}u)(p):=\sum_{x\in\Integer^d}u(x){\rm e}^{-ip\cdot x},\quad u\in l^2(\Integer^d),
\end{align}
where $p\in [-\pi,\pi]^d$ and $p\cdot x=p_0x_0+\dots+p_{d-1}x_{d-1}$. The series
in the right hand side converges in $L^2([-\pi,\pi]^d, dp/(2\pi)^d)$ norm. From the theory
of Fourier series (see, for example, \cite{Kuroda:1980}), $\mathcal{F}$ is a unitary operator with the inverse
\begin{align}
(\mathcal{F}^{-1}u)(x)=\int_{[-\pi,\pi]^d} \frac{dp}{(2\pi)^d}\,u(p){\rm e}^{ip\cdot x},\quad u\in L^2([-\pi,\pi]^d, dp/(2\pi)^d).
\end{align}
In the momentum space, $\square$ is a multiplication operator by the real-valued function
\begin{align}\label{laplacian kernel}
[-\pi,\pi]^d\ni p\mapsto\Box(p)
&=\frac{2}{R^2}\left\{1-\frac{1-\sum_\mu(1- \cos p_\mu)}{\sqrt{\sum_\mu \sin^2 p_\mu+\big[1-\sum_\mu(1- \cos p_\mu)\big]^2}}\right\}\no\\
&=\frac{2}{R^2}\left\{1+\frac{b(\bvec p)-\cos p_0}{\sqrt{a(\bvec p)-2b(\bvec p)\cos p_0}}\right\},
\end{align}
where 
\begin{align}
a(\bvec p)&=1+\sum_{j=1}^{d-1} \sin^2 p_j+b(\bvec p)^2,\label{a}\\
b(\bvec p)&=\sum_{j=1}^{d-1}(1- \cos p_j).\label{b}
\end{align}
We set the above
 unimportant factor $2/R^2$ to be unity by choosing $R=\sqrt{2}$ to make equations simple, and
denote a multiplication operator by a function $p\mapsto\square(p)$ by the same symbol $\square(p)$.
Since the spectrum and the eigenvalues of an operator in a Hilbert space
 is invariant under unitary transformation,
we find that the spectrum of $\square$ is given by
\begin{align}\label{eq: spectrum of square}
\sigma(\square)=\sigma(\mathcal{F}^{-1}\square\mathcal{F})=\sigma(\square(p))=\overline{
\{ \square(p)\,:\, p\in [-\pi,\pi]^d\}},
\end{align}
and the set of eigenvalues of $\square$, $\sigma_{\text{p}}(\square)$, is given by
\begin{align}\label{eq: eigenvalue of square}
\sigma_\text{p}(\square)&=\sigma_\text{p}(\square(p))\no\\
&=\{\lambda\in\Complex \,:\, | \square^{-1}(\{\lambda\}) | =0 \},
\end{align}
where for a
Lebesgue measurable set $A\subset \Real $, $|A|$ is its Lebesgue measure.
\begin{Prop}\label{prop: bdd and injective}
$\square$ is bounded and injective.
\end{Prop}
\begin{proof}
By \eqref{eq: spectrum of square} and \eqref{eq: eigenvalue of square}, 
it suffices to show that (i) $\square(p)$ is a bounded function and
(ii) $\square(p)\not=0$ almost everywhere in $p$ with respect to the Lebesgue measure.

We denote $a=a(\bvec p)$ and $b=b(\bvec p)$ for notational simplicity. By the definitions \eqref{a} and
\eqref{b},
we find
\begin{align}
a\ge 1,\quad b\ge 0,\quad a > 2b.
\end{align}
The last inequality follows since
\begin{align}
a-2b&=1+\sum_j\sin^2 p_j+b^2 -2b \no\\
	&=(b-1)^2 + \sum_j \sin^2 p_j \ge 0,
\end{align}
and the equality holds when and only when
\begin{align}
b=1\quad\text{and}\quad  \sin p_j = 0,\,(j=1,2,\dots,d-1),
\end{align}
which is impossible because whenever the latter is valid, $b=\sum_j(1-\cos p_j)$
has to be even. Thus, $a-2b>0$ for all $\bvec p\in [\pi,\pi]^{d-1}$.

Since the function
\begin{align}
\bvec p\mapsto a(\bvec p)-2b(\bvec p),
\end{align} 
is continuous, it must have a positive minimum $c>0$.
This implies the denominator of $\square (p)$ is bounded from below
by some positive constant :
\begin{align}
a-2b\cos p_0 \ge a-2b \ge c > 0. 
\end{align}
Thus, $p\mapsto\square(p)$ is continuous on $[\pi,\pi]^d$ and its range is compact, which proves (i).

To prove (ii), suppose $\square(p)=0$, that is,
\begin{align}\label{eq: zero of square}
\sqrt{a-2b\cos p_0} = -b+\cos p_0.
\end{align}
This is equivalent to 
\begin{align}\label{eq: zero conditions}
a-b^2=\cos^2 p_0\quad \text{and} \quad b\le \cos p_0.
\end{align}
The first condition in \eqref{eq: zero conditions}, $a-b^2=\cos^2 p_0$, is equivalent to
\begin{align}
1+\sum_j \sin^2 p_j = \cos^2 p_0,
\end{align}
which is possible only when both sides are equal to $1$ since
\begin{align}
1+\sum_j \sin^2 p_j \ge 1, \quad \cos^2 p_0 \le 1.
\end{align}
On the other hand, the second condition in \eqref{eq: zero conditions}, $b\le \cos p_0$, 
is possible only when $\cos p_0\ge 0$ since
$b\ge 0$. 
Therefore \eqref{eq: zero of square} is equivalent to
\begin{align}
\sin p_j =0, \,(j=1,2,\dots,d-1)\quad \text{and} \quad \cos p_0 =1.
\end{align}
Noting that the former implies $b$ must be even, we conclude this condition is valid if and only
if
\begin{align}
p_\mu=0 , \quad \mu=0,1,\dots,d-1,
\end{align}
which means $\square(p)=0$ if and only if $p=0$.
Thus, (ii) is proved.
\end{proof}

By Proposition \ref{prop: bdd and injective} $\square$ has the inverse operator $\square^{-1}$.
The operator $\square^{-1}$ is an unbounded self-adjoint operator, which is strictly positive.
This can be seen by noting that
$\mathcal{F}^{-1}\square^{-1}\mathcal{F}$ is a multiplication operator by the function
\begin{align}
p\mapsto \frac{1}{\square(p)}.
\end{align}
Furthermore, by noting that $\square(p)$ behaves like $\sim \sum_\mu p_\mu^2$ when
$p$ is small, 
we can prove :
\begin{Prop}\label{prop: inverse is L^1}
The function $p\mapsto\square(p)^{-1}$ belongs to $L^1([-\pi,\pi]^d,dp/(2\pi)^d)$ if
$d\ge 3$.
\end{Prop}
\begin{proof}
Fix sufficiently small $\epsilon>0$, and suppose $|p|<\epsilon$. Here, we have employed the notation
\begin{align}
|p|:= \sqrt{\sum_\mu p_\mu^2 },\quad p \in [-\pi,\pi]^d.
\end{align}
Then, since
\begin{align}
\frac{\sum_\mu \sin^2 p_\mu}{\big[1-\sum_\mu(1- \cos p_\mu)\big]^2}
\ge \sum_\mu \sin^2 p_\mu 
\ge |p|^2,
\end{align}
we can estimate
\begin{align}
\square(p) &\ge 1-\frac{1}{\sqrt{1+\frac{\sum_\mu \sin^2 p_\mu}{\big[1-\sum_\mu(1- \cos p_\mu)\big]^2}}}\no\\
	&\ge 1- \frac{1}{\sqrt{1+|p|^2}}.
\end{align}
From the mean value theorem, there is some constant $C_\epsilon>0$ such that
\begin{align}
(1+|p|^2)^{-1/2} = 1 - C_\epsilon |p|^2, \quad |p|<\epsilon.
\end{align}
Thus we obtain
\begin{align}
\square(p)\ge C_\epsilon |p|^2 \quad |p|<\epsilon.
\end{align}
Therefore, if $d\ge 3$,
\begin{align}
\int_{[-\pi,\pi]^d} \frac{dp}{(2\pi)^d}\, \frac{1}{\square(p)}
&=\left(\int_{|p|<\epsilon}+\int_{|p|\ge\epsilon}\right) \frac{dp}{(2\pi)^d}\, \frac{1}{\square(p)}\no\\
&\le \int_{|p|<\epsilon} \frac{dp}{(2\pi)^d}\,\frac{1}{C_\epsilon |p|^2} + \int_{|p|\ge\epsilon}\frac{dp}{(2\pi)^d}\, \frac{1}{\square(p)} \no\\
&<\infty.
\end{align}
This completes the proof.
\end{proof}

\subsection{Lattice overlap boson system}
Let $u,v\in D(\square^{-1})$ and define the inner product of $u,v$ by
\begin{align}\label{inpro h}
(u,v)_h:=(u,\square^{-1} v)_{l^2(\Integer^d)},
\end{align} 
and $D(\square^{-1})$ becomes a Hilbert space with respect to the inner product \eqref{inpro h}.
We denote this Hilbert space by $h$. In the momentum space,
it can be written more explicitly as
\begin{align}
h=\left\{ u\in l^2(\Integer^d)\,:\, \int_{[-\pi,\pi]^d}\frac{dp}{(2\pi)^d}\,
\frac{|(\mathcal{F}u)(p)|^2}{\square(p)} < \infty\right\}.  
\end{align}
In the following, we use the real version of this Hilbert space
\begin{align}\label{h real ft}
h_\Real=\left\{ u\in l^2(\Integer^d;\Real)\,:\, \int_{[-\pi,\pi]^d}\frac{dp}{(2\pi)^d}\,
\frac{|(\mathcal{F}u)(p)|^2}{\square(p)} < \infty\right\}.  
\end{align}
Note that $\square$ and $\square^{-1}$ can be considered to be operators in $h_\Real$.  

The lattice overlap boson field is a Gaussian random process 
$\{\phi(u)\}_{u\in h_\Real}$ labeled by $h_\Real$, that is, $\{\phi(u)\}_{u\in h_\Real}$
is a family of random variables on some probability space $(\Omega,\Sigma,\mu)$ 
satisfying
\begin{itemize}
\item[(i)] $u\mapsto \phi(u)$ is $\Real$-linear almost everywhere in $\mu$, that is,
\begin{align}
\phi(\alpha u+\beta v)=\alpha\phi(u)+\beta\phi(v),\quad\mu\text{-a.e.}
\quad \alpha,\beta\in\Real,\,u,v\in h_\Real.
\end{align}
\item[(ii)] $\{\phi(u)\}_{u\in h_\Real}$ is full, that is, $\Sigma$ is generated by $\{\phi(u)\}_{u\in h_\Real}$.
\item[(iii)] For all $k\in\Real$,
\begin{align}\label{eq: GRP cha func}
\int_\Omega {\rm e}^{ik\phi(u)}\,d\mu={\rm e}^{-||u||^2 k^2/2}.
\end{align}
\end{itemize}

From Minlos's theorem for $\mathscr{S}_\Real'(\Integer^d)$ (see Ref. \cite{Ezawa1988xx}, Theorem 8.5.3
or Ref. \cite{Simon1979} Theorem 2.2) 
the topological dual space of rapidly decreasing real-valued function on $\Integer^d$,
the probability space
$(\Omega,\Sigma,\mu)$ realizing this Gaussian random process can be chosen so that
\begin{align}
\Omega=\mathscr{S}_\Real'(\Integer^d),
\end{align}
and $\Sigma$ is the $\sigma$ algebra generated by the cylinder sets.
Define for $x\in\Integer^d$
\begin{align}
e_x(y):=\delta_{xy},\quad y\in\Integer^d.
\end{align}
Clearly, $e_x \in\mathscr{S}_\Real(\Integer^d)$ for all $x\in\Integer^d$, and
$\{e_x\}_{x\in\Integer^d}$ forms a complete orthonormal system in $l^2(\Integer^d$).
Its Fourier transformation is
\begin{align}
(\mathcal{F}e_x)(p)={\rm e}^{-ip\cdot x}.
\end{align}
Suppose $d\ge 3$, and then by Proposition \ref{prop: inverse is L^1}, $e_x\in h_\Real$. 
Thus, in this case, 
a generator $\phi(x)$ of $\mathfrak{A}_L$
can be regarded as a random variable on $(\mathscr{S'(\Integer^d)},\Sigma,\mu)$ by
\begin{align}\label{field as grv}
\phi(x)(T)=T(e_x),\quad T\in\mathscr{S}'(\Integer^d).
\end{align}
For all local polynomial $a\in\mathfrak{A}_L$, we regard $a$ as a random variable
through the relation \eqref{field as grv}. 

The expectation value of the overlap boson system is defined as
\begin{align}
\Expect a := \int_{\mathscr{S}'(\Integer^d)} a(T)\,d\mu(T),\quad a\in\mathfrak{A}_L,
\end{align}
so that
\begin{align}
\Expect{\phi(x_1)\dots\phi(x_n)}:=\int_{\mathscr{S}'(\Integer^d)} \phi(x_1)\dots\phi(x_n)\,d\mu.
\end{align}
Clearly, we have
\begin{align}
\Expect 1 =1.
\end{align}

The two point correlation of Gaussian random process $\Expect{\phi(u)\phi(v)}$ is 
given by the inner product of $u$ and $v$, since,
by \eqref{eq: GRP cha func}, 
\begin{align}
\Expect{\phi(u)\phi(v)}&=\int_{\mathscr{S}_\Real'(\Integer^d)} \phi(u)\phi(v)\,d\mu \no\\
	&=\frac{1}{i^2}\int_{\mathscr{S}_\Real'(\Integer^d)}
	 \frac{\partial^2}{\partial k\partial l} {\rm e}^{ik\phi(u)+il\phi(v)}d\mu\Big|_{k=l=0} \no\\
	&=\frac{1}{i^2}\frac{\partial^2}{\partial k\partial l}\int_{\mathscr{S}_\Real'(\Integer^d)}
	 {\rm e}^{ik\phi(u)+il\phi(v)} d\mu\Big|_{k=l=0}\no\\
	&=\frac{1}{i^2}\frac{\partial^2}{\partial k\partial l} {\rm e}^{-k^2||u||^2/2 -l^2 ||v||^2/2 -kl(u,v)}\Big|_{k=l=0}\no \\
	&=(u,v)_{h_\Real},
\end{align}
and, in particular, the two point function is given by
\begin{align}
\Expect{\phi(x)\phi(y)}=(e_x,e_y)_{h_\Real}&=(e_x, \square^{-1}e_y)_{l^2(\Integer^d)}\no\\
	&=:\square^{-1}(x,y).
\end{align}
The essential property of the Gaussian random process 
is that the $n$ point correlation functions are completely determined by
two point correlation (Wick's theorem) by
\begin{align}\label{eq: Wick}
\Expect{\phi(u_1)\dots\phi(u_{n})}=
\begin{cases} 0\quad (n=2m+1) \\
\sum_{\text{comb}}\Expect{\phi(u_{j_1})\phi(u_{k_1})}\dots\Expect{\phi(u_{j_m})\phi(u_{k_m})} \quad (n=2m),
\end{cases}		
\end{align}
where $\sum_{\text{comb}}$ means the summation over all $\{j_1,k_1,\dots,j_m,k_m\}$'s
with
\begin{align}
\{1,2,\dots,2m\}=\{j_1,k_1,\dots,j_m,k_m\},\no\\
1\le j_1<\dots < j_m,\quad j_1<k_1,\dots,j_m<k_m,\no
\end{align}
summation of $(2m)!/2^m m!$ terms. Especially, 
\begin{align}\label{eq: Wick special}
\Expect{\phi(x_1)\dots\phi(x_{n})}=
\begin{cases} 0\quad (n=2m+1) \\
\sum_{\text{comb}}\square^{-1}(x_{j_1},x_{k_1})\dots\square^{-1}(x_{j_m},x_{k_m}) \quad (n=2m).
\end{cases}		
\end{align}
 
We call the lattice system $(\mathfrak{A}_L,\Expect\cdot)$ defined above a free overlap boson system.

\subsection{Violation of reflection positivity}
We discuss the quantum mechanical properties of the overlap boson system along the discussion in 
section \ref{sec: lattice QM}.
\begin{Thm}\label{thm: free overlap boson properties}
A free overlap boson system satisfies the condition \textbf{(A1)},\textbf{(A2)},\textbf{(A4)} if $d\ge3$.
\end{Thm}
\begin{proof}
\textbf{(A2)} : Thanks to \eqref{eq: Wick} and \eqref{eq: Wick special}, it suffices to show that
$\square^{-1}(x,y)$ depends only upon $x-y$. This follows 
by the following simple computation
\begin{align}\label{eq: FT of inverse square}
\square^{-1}(x,y)&=\int_{[-\pi,\pi]^d} \frac{dp}{(2\pi)^d}\,\frac{e^{ip\cdot x}{\rm e}^{-ip\cdot y}}{\square(p)}\no\\
	&=\int_{[-\pi,\pi]^d} \frac{dp}{(2\pi)^d}\,\frac{e^{ip\cdot (x-y)}}{\square(p)}.
\end{align} 
Considering this result, we denote $\square^{-1}(x,y)$ also by $\square^{-1}(x-y)$. 

\textbf{(A1)} : Again, from \eqref{eq: Wick} and \eqref{eq: Wick special}, it suffices to show
\begin{align}
\Expect{\theta(\phi(x)\phi(y))}=\overline{\Expect{\phi(x)\phi(y)}}.
\end{align}
But this follows since
\begin{align}
\Expect{\theta(\phi(x)\phi(y))}&=\square^{-1}(\theta x-\theta y)\no\\
	&=\int_{[-\pi,\pi]^d} \frac{dp}{(2\pi)^d}\,\frac{e^{ip\cdot (\theta x-\theta y)}}{\square(p)}\no\\
	&=\int_{[-\pi,\pi]^d} \frac{dp}{(2\pi)^d}\,\frac{e^{ip\cdot (-x_0+y_0)}\,e^{i\bvec p\cdot(\bvec x-\bvec y)}}{\square(p)}\no\\
	&=\int_{[-\pi,\pi]^d} \frac{dp}{(2\pi)^d}\,\frac{{\rm e}^{-ip\cdot (x_0-y_0)}\,{\rm e}^{-i\bvec p\cdot(\bvec x-\bvec y)}}{\square(p_0,-\bvec p)}\no\\
	&=\int_{[-\pi,\pi]^d} \frac{dp}{(2\pi)^d}\,\frac{{\rm e}^{-ip\cdot (x_0-y_0)}\,{\rm e}^{-i\bvec p\cdot(\bvec x-\bvec y)}}{\square(p_0,\bvec p)}\no\\
	&=\overline{\Expect{\phi(x)\phi(y)}},
\end{align}
where we have used $\square(p_0,\bvec p)=\square(p_0,-\bvec p)$.

\textbf{(A4)} : Since $d\ge 3$, $p\mapsto\square(p)^{-1}$ belongs to $L^1([-\pi,\pi])$ from
Proposition \ref{prop: inverse is L^1}. Therefore,
$\square^{-1}(x-y)$ is bounded from above by some constant.
Then, \textbf{(A4)} follows from \eqref{eq: Wick}, \eqref{eq: Wick special}. 
\end{proof}

\begin{Thm}
Let $d\ge 3$. The two point function of a lattice overlap boson system is given by
\begin{align}\label{eq: answer of spec of overlap boson}
\Expect{\phi(x)\phi(0)}\Big|_{x_0>0}=
\int_{[0,\infty)\times[-\pi,\pi]^{d-1}}{\rm e}^{-\lambda (x_0-1)-i\bvec p\cdot \bvec x}\,d\rho(\lambda,\bvec p),
\end{align}
with an $\Real$-valued Borel measure $\rho$ supported on $[0,\infty)\times[-\pi,\pi]^{d-1}$,
\begin{align}
d\rho(\lambda,\bvec p)&=2\pi \frac{a-2b\sqrt{a-b^2}}{\sqrt{a-b^2}\sqrt{a-b^2-1}}
\chi_\mathcal{S}(\bvec p)\delta(\lambda-E_0(\bvec p)) \e^{-\lambda}\frac{d\lambda\, d\bvec p}{\pi\,(2\pi)^{d-1}}\no\\
&\qquad\qquad
+\frac{(b-\cosh \lambda)\sqrt{2b\cosh \lambda-a}}{\cosh^2 \lambda-a+b^2}\chi_{[E_1(\bvec p),\infty)}(\lambda)
 \e^{-\lambda}\frac{d \lambda\,d\bvec p}{\pi\,(2\pi)^{d-1}}.
\end{align}
Here, $\mathcal{S}$ is defined by
\begin{align}
\mathcal{S}:=\left\{ \bvec p \in [-\pi,\pi]^{d-1}\,:\, 1+\sum_{k=1}^{d-1}\sin^2 p_k - b(\bvec p)^2 \ge 0 \right\}
\end{align}
and $E_1(\bvec p)$ is defined by
\begin{align}
\cosh E_1(\bvec p) = \frac{a(\bvec p)}{2b(\bvec p)}, \quad E_1(\bvec p)\ge 0,
\end{align}
for $\bvec p\in [-\pi,\pi]^{d-1}\setminus\{0\}$.
\end{Thm}

\begin{proof}
From \eqref{eq: FT of inverse square}, we obtain by Fubini's theorem,
\begin{align}\label{propagator}
\Expect{\phi(x)\phi(y)}&=\int\frac{dp}{(2\pi)^d}\frac{{\rm e}^{ip(x-y)}}{\square(p)}\no\\
&=\int\frac{d\bvec p}{(2\pi)^{d-1}}{\rm e}^{i\bvec p (\bvec x-\bvec y)}\int\frac{dp_0}{2\pi}
\frac{{\rm e}^{ip_0 (x_0-y_0)}}{\Box(p_0,\bvec p)}.
\end{align}
In the following analysis, we will apply Cauchy's theorem to the $p_0$ integration
\begin{align}
\int\frac{dp_0}{2\pi}
\frac{{\rm e}^{ip_0 (x_0-y_0)}}{\Box(p_0,\bvec p)}.
\end{align}
Note that although the integration \eqref{propagator}
exists since $d\ge 3$, the integrand function of $\bvec p$ in \eqref{propagator}:
\begin{align}\label{p_0 func}
\bvec p \mapsto \int\frac{dp_0}{2\pi}
\frac{{\rm e}^{ip_0 (x_0-y_0)}}{\Box(p_0,\bvec p)}
\end{align}
is not defined at $\bvec p =0$, because at $\bvec p=0$ this $p_0$ integration does not exist.
However, since $\{ \bvec p\in [-\pi,\pi]^{d-1} \,:\, \bvec p =0 \}$ is zero measure with respect to $d\bvec p$ and
does not contribute the $\bvec p$ integration, we may assume $\bvec p\not=0$ in \eqref{p_0 func} when applying
 Cauchy's theorem.

Fix $\bvec p\in [-\pi,\pi]^{d-1}\setminus\{0\}$ and define a complex function $z\mapsto f(z)$ as
\begin{align}
f(z)=\frac{1}{\Box(p_0=z,\bvec p)}.
\end{align}
From \eqref{laplacian kernel}, one finds
\begin{align}
f(z)=\frac{\sqrt{a-2b\cos z}}{\sqrt{a-2b\cos z}+b-\cos z},
\end{align}
where, for notational simplicity, we have written $a=a(\bvec p)$ and $b=b(\bvec p)$ again. 
Here, we have to clarify the meaning of the square root of complex variables.
We define the square root of 
\begin{align}
z=|z| {\rm e}^{i\theta}, \quad\theta\in(-\pi,\pi),
\end{align}
as 
\begin{align}
\sqrt{z}:=\sqrt{|z|} {\rm e}^{i\theta/2},\quad \theta\in(-\pi,\pi).
\end{align}
Namely, we choose the branch where ${\rm Re}\,\sqrt{z}\ge 0$. 

We will investigate the analytic structure of $f$.

First, since the square root function is not continuous when the argument varies across the negative real axis, 
$f$ is not analytic where
\begin{align}\label{branch cut}
a-2b\cos z <0.
\end{align}
To find more explicit condition which is equivalent to \eqref{branch cut},
put $z=x+iy\,(x,y\in\Real)$. Since $\cos(x+iy)=\cos x\cosh y-i\sin x\sinh y$, 
\eqref{branch cut} is equivalent to
\begin{align}
\cos x\cosh y-i\sin x\sinh y>\frac{a}{2b},
\end{align}
which holds true when and only when
\begin{align}\label{branch cut 2}
\cos x\cosh y>\frac{a}{2b}, \quad \sin x\sinh y=0.
\end{align}
The second condition is equivalent to 
\begin{align}
x=n\pi\;(n\in\Integer)\quad\text{or} \quad y=0, 
\end{align}
but the second choice $y=0$ is impossible because
in this case the first condition of \eqref{branch cut 2}
becomes
\begin{align}
\cos x > \frac{a}{2b} \ge 1,
\end{align}
which is never true for real $x$. Therefore, \eqref{branch cut 2}
is equivalent to $x=n\pi$, $n\in\Integer$, and
\begin{align}
\cos (n\pi) \cosh y = (-1)^n \cosh y >\frac{a}{2b},\quad (n\in\Integer).
\end{align}
Hence the condition \eqref{branch cut} occurs when and only when
\begin{align}
x=2n\pi\;(n\in\Integer),
\end{align}
and
\begin{align}
y<-E_1\quad \text{or}\quad E_1<y.
\end{align}

Next, we investigate pole type singularities of $f$ which may appear where its
denominator
\begin{align}
g(z):=\sqrt{a-2b\cos z}+b-\cos z
\end{align}
vanishes. To find a \textit{necessary} condition for $g(z)=0$, let us assume $g(z)=0$.
Then, by taking the square of the both sides of
\begin{align}
\sqrt{a-2b\cos z}=-b+\cos z,
\end{align}
one finds
\begin{align}
\cos^2 z=a-b^2.
\end{align}
Using the identity $\cos^2 z=(1+\cos 2z)/2$ and putting $z=x+iy \,(x,y\in\Real)$ again, one arrives at
\begin{align}\label{cos cosh}
\cos 2x\cosh 2y = 2a-2b^2-1
\end{align} 
and 
\begin{align}\label{sin sinh}
\sin 2x\sinh 2y =0.
\end{align}
Eq. \eqref{sin sinh} implies 
\begin{align}\label{cases}
y=0 \quad\text{or}\quad 2x=n\pi\;(n\in\Integer),
\end{align}
and we will consider both cases respectively.

In the first case, $y=0$, \eqref{cos cosh} becomes
\begin{align}\label{cos 2x}
\cos 2x =2a-2b^2-1 \ge 1,
\end{align}
which is valid only when $2a-2b^2-1=1$, or equivalently,
\begin{align}
\sum_{k=1}^{d-1} \sin^2 p_k =0.
\end{align}
Therefore, this case $y=0$ occurs only when the spacial momentum satisfies
\begin{align}\label{cond for spacial p}
\bvec p=(m_1\pi,\dots,m_{d-1}\pi),\quad m_1,\dots,m_{d-1}\in\{0,1,-1\}.
\end{align} 
Since $\bvec p\not=0$, at least one of the
$m_k$'s ($k=1,\dots,d-1$) should be non-zero.
If \eqref{cond for spacial p} is satisfied, the right hand side of \eqref{cos 2x} becomes $1$ and 
\eqref{cos 2x} implies
\begin{align}
x=n\pi,\quad n\in\Integer.
\end{align}
But, this condition, $y=0$ and $x=n\pi\,(n\in\Integer)$, is not sufficient for $g(z)=0$. In fact, for $n\in\Integer$,
\begin{align}
g(n\pi)&=\sqrt{a-2(-1)^n b}+b-(-1)^n \no\\
 &=|b-(-1)^n|+b-(-1)^n \no\\
 &\ge2
\end{align} 
because
\begin{align}
b=b(\bvec p)=\sum_{k=1}^{d-1}(1-\cos p_k)\ge 2,
\end{align}
due to the fact that at least one of $\cos p_k$'s is equal to $-1$.

In the second case of \eqref{cases}, $2x=n\pi\;(n\in\Integer)$, \eqref{cos cosh}
becomes
\begin{align}
(-1)^n\cosh 2y = 2a-2b^2-1.
\end{align}
And then, this implies that $n$ should be even and 
\begin{align}
\cosh 2y =2a-2b^2-1.
\end{align}
Define $E_0>0$ by\footnote{
Since we assume $\bvec p\not=0$, $\sqrt{a-b^2}$ is less than $1$,
 so that $E_0>0$.
} 
\begin{align}
E_0 &= \frac{1}{2} \cosh^{-1}(2a-2b^2-1)\no\\
&=\cosh^{-1}\sqrt{a-b^2},
\end{align}
and we obtain as a necessary condition for $g(z)=0$, $z=n\pi\pm iE_0\;(n\in\Integer)$.
To find  a sufficient condition for $g(z)=0$, let us assume $z=n\pi\pm iE_0\;(n\in\Integer)$
conversely. Then, we have
\begin{align}
g(n\pi\pm iE_0)&=\sqrt{a-2b(-1)^n\cosh E_0}+b-\cosh E_0\no\\
&=|b-(-1)^n\sqrt{a-b^2}|+ b-\sqrt{a-b^2}\no\\
&=\begin{cases}
0 \quad (\text{if}\;\; b-(-1)^n\sqrt{a-b^2}\le 0)\\
\\
2b\quad (\text{if}\;\;b-(-1)^n\sqrt{a-b^2}\ge 0 )
\end{cases}.
\end{align}
Hence, the necessary and sufficient condition for 
$g(n\pi\pm iE_0)=0$ is that $n$ should be even and
\begin{align}
b-\sqrt{a-b^2}\le 0,
\end{align}
which is equivalent to
\begin{align}\label{S cond}
1+\sum_{k=1}^{d-1} \sin^2 p_k -b(\bvec p)^2\ge 0,
\end{align}
namely, $\bvec p\in\mathcal{S}$.

We now have found all the zeros of the function $g$ : 
\begin{align}
z=2n\pi \pm iE_0, \quad (n\in\Integer, \,\bvec p \in \mathcal{S}).
\end{align}

For a moment, let us assume that spacial momentum $\bvec p$ satisfies
\begin{align}
b-\sqrt{a-b^2}< 0.
\end{align}
In this case, $z=z_n^\pm:=2n\pi \pm iE_0$ is
a simple pole of $f$, as will be seen. From the above argument, $f$ is 
analytic on 
\begin{align}
\Complex\setminus\bigcup_{n\in\Integer}\Big(\{z_n^\pm\}\cup\{2n\pi+iy\,;\,y<-E_1,\,E_1<y\}\Big).
\end{align}
Expand $g$ in Taylor series around $z_n^\pm$ :
\begin{align}
g(z)=g(z_n^\pm)+g'(z_n^\pm)(z-z_n^\pm)+\mathcal{O}\big({(z-z_n^\pm)}^2\big),
\end{align}
on $|z-z_n^\pm|<r$ for sufficiently small $r>0$, and we obtain
\begin{align}
f(z)&=\frac{\sqrt{a-2b\cos z}}{g(z_n^\pm)+g'(z_n^\pm)(z-z_n^\pm)+\mathcal{O}\big({(z-z_n^\pm)}^2\big)}\no\\
&=\frac{\sqrt{a-2b\cos z}}{g'(z_n^\pm)(z-z_n^\pm)+\mathcal{O}\big({(z-z_n^\pm)}^2\big)},
\end{align}
on $|z-z_n^\pm|<r$. Then, we find
\begin{align}
z\mapsto(z-z_n^\pm)f(z)=\frac{\sqrt{a-2b\cos z}}{g'(z_n^\pm)+\mathcal{O}\big(z-z_n^\pm\big)}
\end{align}
is analytic on $|z-z_n^\pm|<r$, and then $z=z_n^\pm$ are simple poles of $z\mapsto \e^{izx_0}f(z)$ with residues
\begin{align}\label{res}
\text{Res}(\e^{izx_0}f(z);z=z_n^\pm)&=(z-z_n^\pm)\e^{izx_0}f(z)\Big|_{z=z_n^\pm}\no\\
&=\frac{\sqrt{a-2b\cos (2n\pi \pm iE_0)}}{g'(2n\pi \pm iE_0)}\e^{\mp E_0 x_0}\no\\
&=\pm\frac{a-2b\sqrt{a-b^2}}{i\sqrt{a-b^2}\sqrt{a-b^2-1}}\e^{\mp E_0 x_0}.
\end{align}
\begin{figure}[!b]
\begin{center}
\includegraphics[scale=0.55,trim=100 200 100 250]{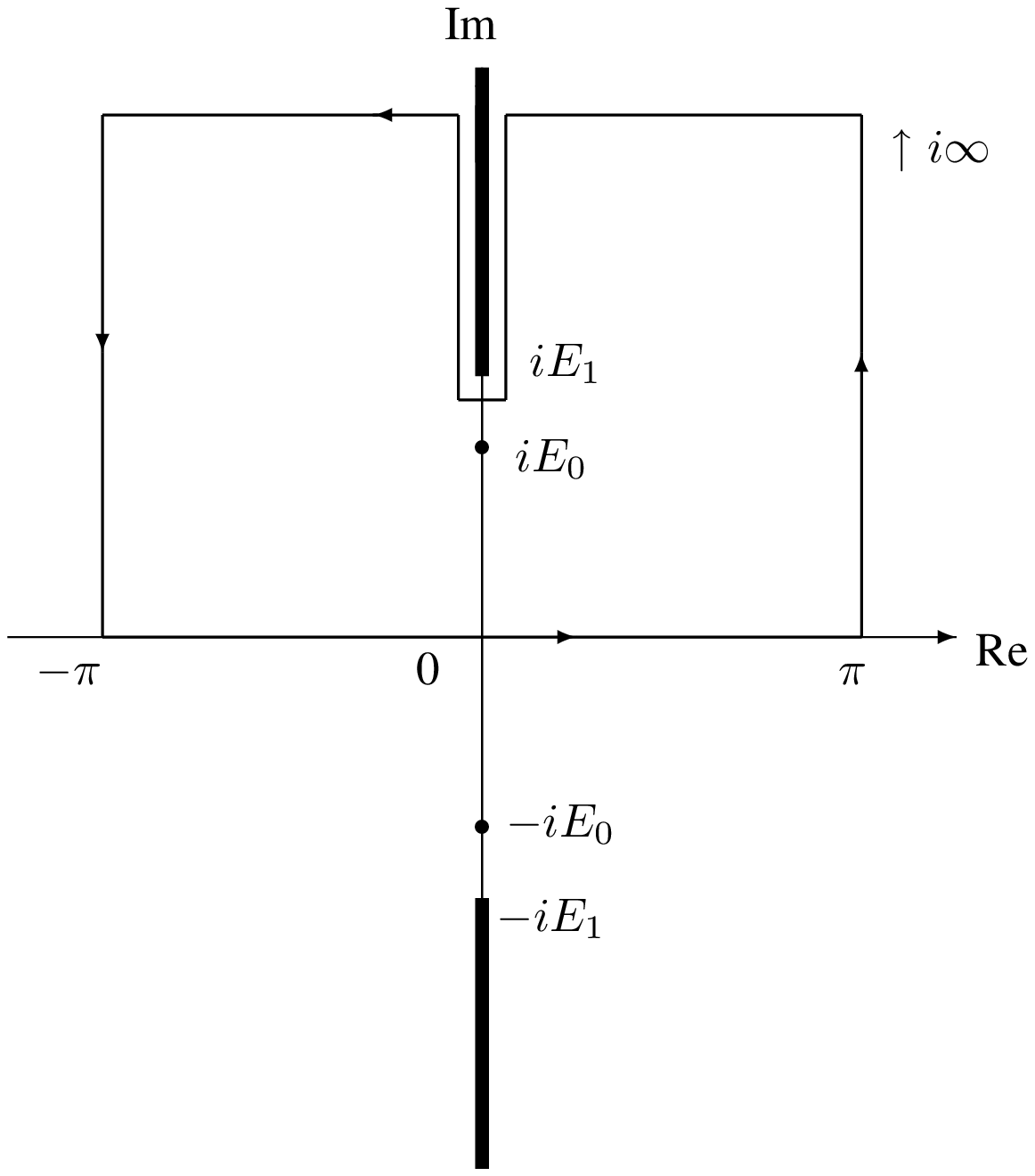} 
\caption{Integration contour.}
\label{fig:p0-integration-contour}
\end{center}
\end{figure}

Applying Cauchy's theorem on the contour drawn in Fig.\ref{fig:p0-integration-contour},
we obtain, for $x_0>0$,
\begin{align}\label{Cauchy}
&\int_{[-\pi,\pi]}\frac{dp_0}{2\pi}\,{\rm e}^{ip_0 x_0}
\frac{1}{\Box(p_0,\bvec p)}\no\\
=&
\int_{[-\pi,\pi]}\frac{dz}{2\pi}\,{\rm e}^{izx_0}f(z)\no\\
=&
2\pi i \,\text{Res}(\e^{izx_0}f(z);z=E_0)-\Bigg(\int_{i\infty+0}^{iE_1+0}+\int_{i\infty-0}^{iE_1-0}\Bigg)\frac{dz}{2\pi}\,{\rm e}^{iz x_0}f(z).
\end{align} 
Recalling our definition of the square root, one finds
\begin{align}
\sqrt{a-2b\cos(iE\pm 0)}=\pm i\sqrt{2b\cosh E-a}.
\end{align}
Then, the integrations of the second term in \eqref{Cauchy} is computed
by putting $z=i\lambda\pm 0$ to become
\begin{align}\label{edge int}
&\Bigg(\int_{i\infty+0}^{iE_1+0}+\int_{i\infty-0}^{iE_1-0}\Bigg){\rm e}^{-\lambda x_0}f(z)\no\\
=&\int_{E_1}^\infty\frac{i d\lambda}{2\pi}{\rm e}^{-\lambda x_0}\Big(f(i\lambda+0)+f(i\lambda-0)\Big)\no\\
=-&\int_{E_1}^\infty\frac{d\lambda}{\pi}{\rm e}^{-\lambda x_0}\frac{(b-\cosh \lambda)\sqrt{2b\cosh \lambda-a}}{\cosh^2 \lambda-a+b^2}.
\end{align}
By substituting \eqref{res} and \eqref{edge int} into \eqref{Cauchy}, we arrive at
\begin{align}\label{answer of spec}
\int_{[-\pi,\pi]}\frac{dp_0}{2\pi}\,
\frac{{\rm e}^{ip_0 x_0}}{\Box(p_0,\bvec p)}
&=2\pi \frac{a-2b\sqrt{a-b^2}}{\sqrt{a-b^2}\sqrt{a-b^2-1}}\e^{-E_0 x_0}\no\\
&\qquad+\int_{E_1}^\infty\frac{d\lambda}{\pi}{\rm e}^{-\lambda x_0}\frac{(b-\cosh \lambda)\sqrt{2b\cosh \lambda-a}}{\cosh^2 \lambda-a+b^2}.
\end{align}

Considering the case where spacial momentum $\bvec p$ satisfies
\begin{align}
b-\sqrt{a-b^2}\ge 0,
\end{align}
we find that there is no pole term and only the second term of \eqref{answer of spec} survives.
Note that, in the case of equality, even though $g(z)=0$, $f$ has no isolated pole.
In this case, the numerator of $f(z)$ also vanishes and $E_0=E_1$.

From \eqref{answer of spec}, we learn
\begin{align}
\Expect{\phi(x)\phi(0)}=\int_{[0,\infty)\times [-\pi,\pi]^{d-1}} {\rm e}^{-\lambda (x_0-1)-i\bvec p\cdot x}\,
d\rho(\lambda,\bvec p)
\end{align}
with
\begin{align}\label{eq: spectral representation}
d\rho(\lambda,\bvec p)&=2\pi \frac{a-2b\sqrt{a-b^2}}{\sqrt{a-b^2}\sqrt{a-b^2-1}}
\chi_\mathcal{S}(\bvec p)\delta(\lambda-E_0(\bvec p)) \e^{-\lambda} \frac{d\lambda\, d\bvec p}{\pi\,(2\pi)^{d-1}}\no\\
&\qquad\qquad
+\frac{(b-\cosh \lambda)\sqrt{2b\cosh \lambda-a}}{\cosh^2 \lambda-a+b^2}\chi_{[E_1(\bvec p),\infty)}(\lambda)
 \e^{-\lambda}\frac{d \lambda\,d\bvec p}{\pi\,(2\pi)^{d-1}}.
\end{align}
\end{proof}

This expression is clearly the Umezawa-Kamefuchi-K\"all\'en-Lehmann representation of the propagator,
and the spectral density $\sigma$
of the overlap boson system is
given by
\begin{align}
\sigma(\lambda,\bvec p)=\frac{(b-\cosh \lambda)\sqrt{2b\cosh \lambda-a}}{\cosh^2 \lambda-a+b^2}\chi_{[E_1(\bvec p),\infty)}(\lambda) \e^{-\lambda}.
\end{align} 
The important observation is that this function $\sigma$ becomes negative on the set
\begin{align}
\{ (\lambda,\bvec p)\in [0,\infty)\times [-\pi,\pi]^{d-1} \,:\, b(\bvec p) - \cosh\lambda < 0,\quad \lambda\ge E_1(\bvec p)  \},
\end{align} 
which has a positive Lebesgue measure.
This means that the spectral density of the overlap boson system is not non-negative function. 
From Corollary \ref{cor: violation}, one concludes that overlap boson system 
violates at least on of the
conditions from \textbf{(A1)} to \textbf{(A4)}. But from Theorem \ref{thm: free overlap boson properties}, 
the only candidate is
\textbf{(A3)}, the reflection positivity condition. Thus, we finally arrive at
\begin{Thm}
The lattice overlap boson system violates the link reflection positivity condition \textbf{(A3)}.
\end{Thm}
\section{Summary and conclusion}
We have proved that a lattice model which satisfies the assumptions \textbf{(A1)} - \textbf{(A4)}
permits the Euclidean version of Umezawa-Kamefuchi-K\"all\'en-Lehmann representation of 
two point correlation functions with a positive spectral density. This implies that
a lattice model with a spectral density function which is not positive definite violates at least one of
the assumptions \textbf{(A1)} - \textbf{(A4)}.

Lattice overlap boson, which plays an important role when formulating Wess-Zumino model
on the lattice with exact $U(1)_R$ symmetry, has a spectral density function which is not positive. 
Considering that overlap boson fulfills \textbf{(A1)}, \textbf{(A2)}, and \textbf{(A4)}, it follows that
overlap boson violates the condition \textbf{(A3)}, reflection positivity condition.
 
\section*{Ackonwledgements}
The author thanks Y. Kikukawa for fruitful discussions. This work is supported by
JSPS Research Fellowships for Young Scientists and 
by World Premier International Center Initiative (WPI Program), MEXT, Japan.





\bibliographystyle{model1-num-names}
\bibliography{myref}







\end{document}